\documentclass{article} 
\usepackage{iclr2023_conference,times}

\usepackage{hyperref}
\hypersetup{
    colorlinks,
    linkcolor={red!50!black},
    citecolor={blue!50!black},
    urlcolor={blue!80!black}
}

\usepackage{url}
\usepackage[T1]{fontenc}    
\usepackage{hyperref}       
\usepackage{url}            
\usepackage{booktabs}       
\usepackage{amsfonts}       
\usepackage{nicefrac}       
\usepackage{microtype}      
\usepackage{xcolor}         
\usepackage{wrapfig}

\usepackage{microtype}
\usepackage{graphicx}
\usepackage{subfigure}
\usepackage{booktabs} 
\usepackage{amsfonts}
\usepackage{amssymb}
\usepackage{amsmath}
\usepackage{sansmath}
\usepackage{MnSymbol,wasysym}
\usepackage{subfigure}
\usepackage{amsthm}
\usepackage{tcolorbox}
\usepackage{todo}
\newtheorem{theorem}{Theorem}[section]
\newtheorem{definition}[theorem]{Definition}
\newtheorem{assumption}[theorem]{Assumption}
\newtheorem{proposition}[theorem]{Proposition}
\newtheorem{lemma}[theorem]{Lemma}

\newtheorem{remark}[theorem]{Remark}

\usepackage{mathtools}
\usepackage{enumitem}
\newcommand{\eat}[1]{}

\usepackage{hyperref}

\usepackage{bm}
\usepackage{multirow}
\newtcolorbox[list inside=tcb]{mytcb}[1][]{#1}

\DeclareMathOperator{\proj}{proj}
\DeclareMathOperator{\dist2}{dist}

\allowdisplaybreaks[4]
\usepackage{caption}
\iclrfinalcopy
\title{A Convergent Single-Loop Algorithm for Relaxation of Gromov-Wasserstein in Graph Data }
\author{Jiajin Li \\
Stanford University\\
 \texttt{jiajinli@stanford.edu}
\And
Jianheng Tang\\
HKUST (GZ)\\
\texttt{jtangbf@connect.ust.hk}
\And
Lemin Kong\\
CUHK\\
\texttt{lkong@se.cuhk.edu.hk}
\And
Huikang Liu \\
SUFE\\
\texttt{liuhuikang@sufe.edu.cn}
\And
Jia Li\\
HKUST (GZ)\\
\texttt{jialee@ust.hk}
\And 
Anthony Man-Cho So \\
CUHK\\
\texttt{manchoso@se.cuhk.edu.hk}
\And 
Jose Blanchet\\
Stanford University\\
\texttt{jose.blanchet@stanford.edu}
}

\begin{document}

\maketitle
\vspace{-3mm}
\begin{abstract}
In this work, we present the Bregman Alternating Projected Gradient (BAPG) method, a single-loop algorithm that offers an approximate solution to the Gromov-Wasserstein (GW) distance. 
We introduce a novel relaxation technique that balances accuracy and computational efficiency, albeit with some compromises in the feasibility of the coupling map.  Our analysis is based on the observation that the GW problem satisfies the Luo-Tseng error bound condition, which relates to estimating the distance of a point to the critical point set of the GW problem based on the optimality residual.
This observation allows us to provide an approximation bound for the distance between the fixed-point set of BAPG and the critical point set of GW. Moreover, under a mild  technical assumption, we can  show that BAPG converges to its fixed point set.
The effectiveness of BAPG has been validated through comprehensive numerical experiments in graph alignment and partition tasks, where it outperforms existing methods in terms of both solution quality and wall-clock time.
\end{abstract}

\section{Introduction}

The GW distance provides a flexible way to compare and couple probability distributions supported on different metric spaces. This has led to a surge in literature that applies the GW distance to various structural data analysis tasks, including 2D/3D shape  matching~\citep{peyre2016gromov,memoli2004comparing,memoli2009spectral}, molecule analysis~\citep{vayer2018fused,titouan2019optimal}, graph alignment and partition~\citep{chowdhury2019gromov,xu2019gromov,xu2019scalable,chowdhury2021generalized,gao2021unsupervised}, graph embedding and classification~\citep{vincent2021online, xu2022representing}, generative modeling~\citep{bunne2019learning,xu2021learning}.

Although the GW distance has gained a lot of attention in the machine learning and data science communities, most existing algorithms for computing the GW distance are double-loop algorithms that require another iterative algorithm as a subroutine, making them not ideal for practical use.
Recently, an entropy-regularized iterative sinkhorn projection algorithm called eBPG was proposed by \citet{solomon2016entropic}, which has been proven to converge under the Kurdyka-\L{}ojasiewicz framework. However, eBPG has several limitations. Firstly, it addresses an entropic-regularized GW objective, whose regularization parameter has a major impact on the model's performance. Secondly, it requires solving an entropic optimal transport problem at each iteration, which is both computationally expensive and not practical.
In an effort to solve the GW problem directly, \citet{xu2019gromov} proposed the Bregman projected gradient (BPG), which is still a double-loop algorithm that relies on another iterative algorithm as a subroutine. Additionally, it suffers from numerical instability due to the lack of an entropic regularizer.
While \citet{titouan2019optimal,memoli2011gromov} introduced the Frank-Wolfe method to solve the GW problem, they still relied on linear programming solvers and line-search schemes, making it unsuitable for even medium-sized tasks.
Recently, \citet{xu2019gromov} developed a simple heuristic, single-loop method called BPG-S based on BPG that showed good empirical performance on node correspondence tasks. However, its performance in the presence of noise is unknown due to the lack of theoretical support.

The main challenge lies in efficiently tackling the Birkhoff polytope constraints (i.e., the polytope of doubly stochastic matrices) for the coupling matrix. The key issue is that there is no closed update for its Bregman projection, which forces current algorithms to rely on computationally expensive or hyperparameter-sensitive iterative methods.
To address this difficulty, we propose a single-loop algorithm (BAPG) that solves the GW distance approximately. Our solution incorporates a novel relaxation technique that sacrifices some feasibility of the coupling map to achieve computational efficiency. This violation is acceptable for certain learning tasks, such as graph alignment and partition, where the quality of the coupling is not the primary concern.
 We find that BAPG can obtain desirable performance on some graph learning tasks as the performance measure for those tasks is the matching accuracy instead of the sharpness of the probabilistic correspondence.  In conclusion,  BAPG offers a way to sacrifice the feasibility for both computational efficiency and matching accuracy. 

In our approach, we decouple the Birkhoff polytope constraint  into separate simplex constraints for the rows and columns. The projected gradient descent is then performed on a constructed penalty function using an alternating fashion.
By utilizing the closed-form Bregman projection of the simplex constraint with relative entropy as the base function, BAPG only requires matrix-vector/matrix-matrix multiplications and element-wise matrix operations at each iteration, making it a computationally efficient algorithm.
Thus, BAPG has several convenient properties such as compatibility with GPU implementation, robustness with regards to the step size (the only hyperparameter), and low memory requirements.

Next, we investigate the approximation bound and convergence behavior of BAPG. We surprisingly discover that the GW problem satisfies the Luo-Tseng error bound condition~\citep{luo1992error}. This fact allows us to bound the distance between the fixed-point set of BAPG and the critical point set of the GW problem, which is a notable departure from the usual approach of utilizing the Luo-Tseng error bound condition in establishing the  linear convergence rate for structured convex problems~\citep{zhou2017unified}. With this finding, we are able to quantify the approximation bound for the fixed-point set of BAPG explicitly. Moreover, we establish the subsequence convergence result when the accumulative asymmetric error of the Bregman distance is bounded.

Lastly, we present extensive experimental results to validate the effectiveness of BAPG for graph alignment and graph partition. Our results demonstrate that BAPG outperforms other heuristic single-loop and theoretically sound double-loop methods in terms of both computational efficiency and matching accuracy. We also conduct a sensitivity analysis of BAPG and demonstrate the benefits of its GPU acceleration through experiments on both synthetic and real-world datasets. All theoretical insights and results have been well-corroborated in the experiments. 

\section{Proposed Algorithm}
\label{sec:method}

In this section, we begin by presenting the GW distance as a nonconvex quadratic problem with Birkhoff polytope constraints. We then delve into the theoretical insights and computational characteristics of our proposed algorithm, BAPG.

The Gromov-Wasserstein distance was first introduced in \citep{memoli2011gromov,memoli2014gromov,peyre2019computational} as a way to quantify the distance between two probability measures supported on different metric spaces.  More precisely:
\begin{definition}[GW distance]
    \label{defi:gw}
    Suppose that we are given two unregistered compact metric spaces $(\mathcal{X},d_X)$, $(\mathcal{Y},d_Y)$ accompanied with Borel probability measures $\mu,\nu$ respectively. The GW distance between $\mu$ and $\nu$ is defined as 
    \[
    \inf_{\pi \in \Pi(\mu,\nu)}  \iint |d_X(x,x')-d_Y(y,y')|^2 d\pi(x,y)d\pi(x',y'),
    \]
    where $\Pi(\mu,\nu)$ is the set of all probability measures on $\mathcal{X}\times\mathcal{Y}$ with $\mu$ and $\nu$ as marginals. 
 \end{definition}

Intuitively, the GW distance aims to preserve the isometric structure between two probability measures through optimal transport. If there is a map that pairs $x\rightarrow y$ and $x'\rightarrow y'$, then the distance between $x$ and $x'$ should be similar to the distance between $y$ and $y'$. Due to these desirable properties, the GW distance is a powerful tool in structural data analysis, particularly in graph learning. Some examples of its applications include \citep{vayer2019sliced,xu2019gromov,xu2019scalable,solomon2016entropic,peyre2016gromov} and related references.

To start with our algorithmic developments, we consider the discrete case for simplicity and practicality,  where $\mu$ and $\nu$ are two empirical distributions, i.e., $\mu = \sum_{i=1}^n \mu_i \delta_{x_i}$ and  $\nu = \sum_{j=1}^m \nu_j \delta_{y_j}$. As a result, the GW distance can be reformulated as follows:
	\begin{equation}
	\label{eq:gw_qua}
	\begin{aligned}
	 & \min_{\pi \in \mathbb{R}^{n \times m}} -\text{Tr}(D_X\pi D_Y\pi^T)\\
	& \quad \,\text{s.t.} \quad\pi \mathbf{1}_m = \mu, \, \pi^T\mathbf{1}_n = \nu, \,  \pi \ge 0, 
	\end{aligned}
	\end{equation}
where $D_X$ and $D_Y$ are two symmetric distance matrices. 

\subsection{Relaxation of GW Distance}
Now, we will introduce our relaxation of GW distance.  The nonconvex quadratic program \eqref{eq:gw_qua} with polytope constraints is typically addressed by (Bregman) projected gradient descent type algorithms. However, existing algorithms require an inner iterative algorithm, such as Sinkhorn \citep{cuturi2013sinkhorn} or the semi-smooth Newton method \citep{cuturi2013sinkhorn}, to solve the regularized optimal transport problem at each iteration. This can lead to a computationally intensive double-loop scheme, which is not ideal for GPU-friendly computation. To overcome this issue, we aim to handle the row and column constraints separately using an operator splitting-based relaxation technique.

For simplicity, we consider the compact form for better exploiting the problem specific structures:
\begin{equation}
\label{eq:gw_compact}
    \min_{\pi} f(\pi)+g_1(\pi)+g_2(\pi). 
\end{equation}
Here, $f(\pi) =-\text{Tr}(D_X\pi D_Y\pi^T)$ is a nonconvex quadratic function;  $g_1(\pi)=\mathbb{I}_{\{\pi \in C_1\}}$ and 
$g_2(\pi)= \mathbb{I}_{\{\pi \in C_2\}}$ are two indicator functions over closed convex polyhedral sets. Here, $C_1=\{\pi\ge 0: \pi \mathbf{1}_m = \mu\}$ and $C_2=\{\pi\ge 0: \pi^T \mathbf{1}_n = \nu\}$.
 To decouple the Birkhoff polytope constraint, 
we adopt the operator splitting strategy to reformulate \eqref{eq:gw_compact} as
\begin{equation}
\label{eq:gw_bilinear}
\begin{aligned}
	 & \min_{ \pi = w}\, f(\pi,w) +g_1(\pi)+g_2(w)\\
	\end{aligned}
	\end{equation}
where  $f(\pi,w) = -\text{Tr}(D_X\pi D_Yw^T)$.  Then, we penalize the equality constraint and process the alternating minimization scheme on the constructed penalized function, i.e.,
\begin{align*}
& F_\rho(\pi,w) = f(\pi,w)+g_1(\pi)+g_2(w)+\rho D_h(\pi,w).
\end{align*}
Here, $D_h(\cdot,\cdot)$ is the so-called Bregman divergence, i.e.,
$
D_{h}(x, y):=h(x)-h(y)-\langle\nabla h(y), x-y\rangle,
$
where $h(\cdot)$ is the Legendre function, e.g., $\tfrac{1}{2}\|x\|^2$, relative entropy $x\log x$, etc.

\subsection{Bregman Alternating Projected Gradient (BAPG)}
Next, we present the proposed single-loop Bregman alternating projected gradient (BAPG) method. The crux of BPAG is to take the alternating projected gradient descent step between $C_1$ and $C_2$.
For the $k$-th iteration, the BAPG update takes the form
\begin{equation}
\label{eq:bapg_general}
\begin{aligned}
\pi^{k+1} & = \mathop{\arg\min}_{\pi \in C_1}\left\{ f(\pi,w^k) + \rho D_h(\pi, w^k)\right\},  \\
w^{k+1} & = \mathop{\arg\min}_{w\in C_2}\left\{f(\pi^{k+1},w) + \rho D_h(w,\pi^{k+1})\right\}.
\end{aligned}
\end{equation}
The choice of relative entropy as $h$ also brings the advantage of efficient computation of Bregman projection for simplex constraints, such as $C_1$ and $C_2$, as discussed in~\citep{krichene2015efficient}. These observations result in closed-form updates in each iteration of BAPG in \eqref{eq:bapg_general}. We refer to this specific case as KL-BAPG.
\begin{mytcb}[title= KL-BAPG]
\vspace{-1mm}
\begin{equation}
\label{eq:bapg_update_c}
\begin{aligned}
& \pi \leftarrow \pi \odot \exp({D_X \pi D_Y}/{\rho}), \quad \pi \leftarrow  \text{diag}(\mu./\pi\mathbf{1}_m) \pi, \\
&\pi \leftarrow \pi \odot \exp({D_X \pi D_Y}/{\rho}), \quad 
\pi \leftarrow\pi  \text{diag}(\nu./{\pi}^T\mathbf{1}_n),
\end{aligned}
\end{equation}
\end{mytcb}

where $\rho$ is the step size and $\odot$
denotes element-wise (Hadamard) matrix multiplication. 
KL-BAPG has several advantageous properties that make it ideal for medium to large-scale graph learning tasks. Firstly, it is a single-loop algorithm that only requires matrix-vector/matrix-matrix multiplications and element-wise matrix operations, which are highly optimized on GPUs. Secondly, unlike the entropic regularization parameter in eBPG, KL-BAPG is less sensitive to the choice of the step size $\rho$. Thirdly, KL-BAPG only requires one memory operation for a matrix of size $nm$, which is the main bottleneck in large-scale optimal transport problems rather than floating-point computations.~\citep{mai2021fast}.

Similar to the quadratic penalty method \citep{nocedal2006numerical}, BAPG is an infeasible method that only converges to a critical point of \eqref{eq:gw_qua} in an asymptotic sense, meaning there will always be an infeasibility gap if $\rho$ is chosen as a constant. Despite this, BAPG is a suitable option for learning tasks that prioritize efficiency and matching accuracy, such as graph alignment and partition. This idea of sacrificing some feasibility for other benefits is further supported by recent studies such as the relaxed version of GW distance proposed in \citep{vincent2021semi} for graph partitioning.
Additionally, \cite{sejourne2021unbalanced} introduced a closely related marginal relaxation, but they did not develop an efficient algorithm with a convergence guarantee. That is, we make $\pi = w$ and $F_\rho(\pi,\pi)$ is the objective introduced in \citep{sejourne2021unbalanced}.
Our experiments in Sec \ref{sec:graph_align} and \ref{sec:graph_partition} demonstrate that KL-BAPG outperforms existing baselines in graph alignment and partitioning tasks.

\section{Theoretical Results}
\label{sec:conv}

In this section, we present the theoretical results that have been carried out in this paper. This includes the approximation bound of the fixed-point set of BAPG and its convergence analysis. The cornerstone of our analysis is the following regularity condition for the GW problem in equation~\eqref{eq:gw_qua}
\begin{proposition}[Luo-Tseng Error Bound Condition for \eqref{eq:gw_qua}] 
\label{prop:erb}
There exist scalars $\epsilon>0$ and $\tau >0$ such that 
\begin{equation}
\label{eq:erb}
\dist2(\pi,\mathcal{X}) \leq \tau \left\|\pi - \proj_{C_1\cap C_2}(\pi+D_X \pi D_Y)\right\|,
\end{equation}
whenever $\|\pi - \proj_{C_1\cap C_2}(\pi+D_X \pi D_Y)\|\leq \epsilon$, where $\mathcal{X}$ is the critical point set of \eqref{eq:gw_compact} defined by 
\begin{equation}
\label{eq:gw_critical}
    \mathcal{X} = \{\pi \in C_1\cap C_2: 0\in \nabla f(\pi)+\mathcal{N}_{C_1}(\pi)  +\mathcal{N}_{C_2}(\pi)\}
\end{equation}
and $\mathcal{N}_{C}(\pi)$ denotes the normal cone to $C$ at $\pi$. 
\end{proposition}
As the GW problem is a nonconvex quadratic program with polytope constraint, we can invoke Theorem 2.3 in \citep{luo1992error} to conclude that the  error bound condition \eqref{eq:erb} holds on the whole feasible set $C_1 \cap C_2$. Proposition \ref{prop:erb} extends \eqref{eq:erb} to the whole space $\mathbb{R}^{n\times m}$. This regularity condition is trying to bound the distance of any coupling matrix to the critical point set  of the GW problem by its optimality residual, which is characterized by the difference for one step projected gradient descent.
It turns out that this error bound condition  plays an important role in quantifying the approximation bound for the fixed points set of BAPG explicitly. 

\subsection{Approximation Bound for the Fixed-Point Set of BAPG}
To start, we present one key lemma that shall be used in studying the approximation bound of BAPG. 
\begin{lemma}
\label{lem:lemma_proj}
Let $C_{1}$ and $C_{2}$ be convex polyhedral sets. There exists a constant $M>0$ such that
\begin{equation*}
\begin{aligned}
&\left\|\proj_{C_{1}}(x)+\proj_{C_{2}}(y)- 2 \proj_{C_1\cap C_2}\left(\frac{x+y}{2}\right)\right\| 
\leqslant  M\left\|\proj_{C_{1}}(x)-\proj_{C_{2}}(y)\right\|, \quad \forall x\in C_1, y\in C_2.
\end{aligned}
\vspace{-3mm}
    \end{equation*}
\end{lemma}
The proof idea follows essentially from the observation that the inequality can be regarded as the stability of the optimal solution for a linear-quadratic problem, i.e., \begin{equation*}
(p(r),q(r)) = 
\begin{aligned}
 & \mathop{\arg\min}_{p,q} \frac{1}{2}\|x-p\|^2 +\frac{1}{2}\|y-q\|^2\\
 & \,\, \quad \text{s.t.} \, \, \quad  p-q= r, \, p \in C_1, q\in C_2. 
\end{aligned}
\end{equation*}
The parameter $r$ is indeed the perturbation quantity.  If  $r=0$, we have $p(0)= q(0) = \proj_{C_1 \cap C_2} (\frac{x+y}{2})$; by choosing $r = \proj_{C_1}(x) - \proj_{C_2}(y)$, it is easy to see that $(p(r), q(r)) = (\proj_{C_1}(x), \proj_{C_2}(y))$.  Together with 
Theorem 4.1 in \citep{zhang2020global}, the desired result is obtained. All the proof details are given in Appendix.

Equipped with Lemma \ref{lem:lemma_proj} and Proposition \ref{prop:erb}, it is not hard to obtain the  approximation result. 
\begin{proposition} [Approximation Bound of the Fixed-point Set of BAPG]
\label{prop:bapg_approx}
The point $(\pi^\star,w^\star)$ belongs to the fixed-point set $\mathcal{X}_{\text{BAPG}}$ of BAPG if it satisfies 
\begin{equation}
\label{eq:x_bapg}
\begin{aligned}
& \nabla f(w^\star) + \rho (\nabla h(\pi^\star)-\nabla h(w^\star)) + p =0,  \\ 
& \nabla f(\pi^\star) + \rho (\nabla h(w^\star)-\nabla h(\pi^\star)) + q =0, 
\end{aligned}
\end{equation} 
where $p \in \mathcal{N}_{C_1}(\pi^\star)$ and $q \in \mathcal{N}_{C_2}(w^\star)$. Then, 
the infeasibility error satisfies $\|\pi^\star-w^\star\| \leq  \frac{\tau_1}{\rho}$ and the gap between $\mathcal{X}_{\text{BAPG}}$ and $\mathcal{X}$ satisfies
\[
\dist2\left(\frac{\pi^\star+w^\star}{2},\mathcal{X}\right) \leq \frac{\tau_2}{\rho},
\]
where $\tau_1$ and $\tau_2$ are two constants. 
\end{proposition}
\begin{remark}
If $\pi^\star = w^\star$, then $\mathcal{X}_{\text{BAPG}}$ (\eqref{eq:x_bapg}) is identical to $\mathcal{X}$ and BAPG can reach a critical point of the GW problem \eqref{eq:gw_qua}. Proposition \ref{prop:bapg_approx} indicates that as $\rho \rightarrow +\infty$, the infeasibility error term $\|\pi^\star-w^\star\|$ shrinks to zero and thus BAPG converges to a critical point of \eqref{eq:gw_qua} in an asymptotic way. Furthermore, it explicitly quantifies the approximation gap when we select the parameter $\rho$ as a constant. The proof can be found in Appendix. The explicit form of $\tau_1$ and $\tau_2$ only depend on the problem itself, including $\sigma_{\textnormal{max}}(D_X)\sigma_{\textnormal{max}}(D_Y)$, the constant for the Luo-Tseng error bound condition in Proposition \ref{prop:erb} and so on. 
\end{remark}
\subsection{Convergence Analysis of BAPG}
A natural follow-up question is whether BAPG converges. We answer  affirmatively. Under several standard assumptions, we demonstrate that any limit point of BAPG is an element of $\mathcal{X}_{\text{BAPG}}$.
With this goal in mind, we first establish the sufficient decrease property of the potential function $F_\rho(\cdot)$,
\begin{proposition}
\label{prop:basic_bapg}
Let $\{(\pi^k,w^k)\}_{k\ge 0}$ be the sequence generated by BAPG. Suppose that $\sum_{k=0}^\infty \left(D_h(\pi^{k+1},w^k) - D_h(w^k, \pi^{k+1})\right)$ is bounded. Then, we have
    \begin{equation}
    \label{eq:bapg_suff_decre}
    \begin{aligned}
    & F_\rho(\pi^{k+1},w^{k+1}) - F_\rho(\pi^k ,w^k) \leq -  \rho D_h(\pi^k,\pi^{k+1}) - \rho D_h(w^k,w^{k+1}). 
    \end{aligned}
    \end{equation}
\end{proposition}
As $F_\rho(\cdot)$ is coercive,  we have  
$\sum_{k=0}^\infty D_h(\pi^k,\pi^{k+1}) + D_h(w^k,w^{k+1}) < +\infty$. Both $\{D_h(\pi^k,\pi^{k+1})\}_{k \ge 0}$ and $\{D_h(w^k,w^{k+1})\}_{k \ge 0}$ converge to zero. Thus, the following convergence result holds.
\begin{theorem}[Subsequence Convergence of BAPG]
\label{thm:sub_bapg}
Any limit point of the sequence $\{(\pi^{k},w^k)\}_{k \ge 0 }$ generated by BAPG belongs to $\mathcal{X}_{\text{BAPG}}$.
\end{theorem}
\begin{remark}
Verifying the boundedness of the accumulative asymmetric error is a challenging task, except in the case where $h$ is quadratic. To address this we perform empirical verification on a 2D toy example, as described in Sec \ref{sec:toy}. The results of this verification for various step sizes, can be seen in Fig. \ref{fig:accumulative_error} in Appendix \ref{sec:verify}.
Additionally, when $h$ is quadratic, we can employ the Kurdyka-Lojasiewicz analysis framework, which was developed in \citep{attouch2010proximal,attouch2013convergence} to prove global convergence.
\end{remark}

To the best of our knowledge, the convergence analysis of alternating projected gradient descent methods has only been given under the convex setting, see \citep{wang2016stochastic,nedic2011random} for details.
In this paper, by heavily exploiting the error
bound condition of the GW problem,  we take the first step and provide a new path to conduct the analysis of alternating projected descent
method for nonconvex problems, which could be of independent interest.


\vspace{-2mm}
\section{Experiment Results}\label{sec:exp}
\vspace{-2mm}
In this section, we provide extensive experiment results to validate the effectiveness of the proposed KL-BAPG on various representative graph learning tasks, including graph alignment~\citep{tang2023robust} and graph partition~\citep{li2021mask}. 
All simulations are implemented using Python 3.9 on a high-performance computing server running Ubuntu 20.04 with an Intel(R) Xeon(R) Gold 6226R CPU and an NVIDIA GeForce RTX 3090 GPU. For all methods conducted in the experiment part, we use the relative error $\|\pi^{k+1}-\pi^k\|_2 / \|\pi^k\|_2\leq 1e^{-6}$ and the maximum iteration as the stopping criterion, i.e., $\min\{k \in \mathbb{Z}:\|\pi^{k+1}-\pi^k\|_2 / \|\pi^k\|_2\leq 1e^{-6} \, \textnormal{and} \, k\leq 2000 \}$. Our code is available at \href{https://github.com/squareRoot3/Gromov-Wasserstein-for-Graph}{https://github.com/squareRoot3/Gromov-Wasserstein-for-Graph}.
\vspace{-1mm}
\subsection{Toy 2D Matching Problem}\label{sec:toy}
\vspace{-1mm}

In this subsection, we study a toy matching problem in 2D to confirm our theoretical insights and results in Sec \ref{sec:method} and \ref{sec:conv}. Fig. \ref{fig:syn} (a) illustrates an example of mapping a two-dimensional shape without any symmetries to a rotated version of the same shape. Here, we sample 300 points from the source shape and 400 points from the target shape, and use the Euclidean distance to construct the distance matrices $D_X$ and $D_Y$.

\begin{figure*}[!t]
	\centering
	\includegraphics[width=0.9\textwidth]{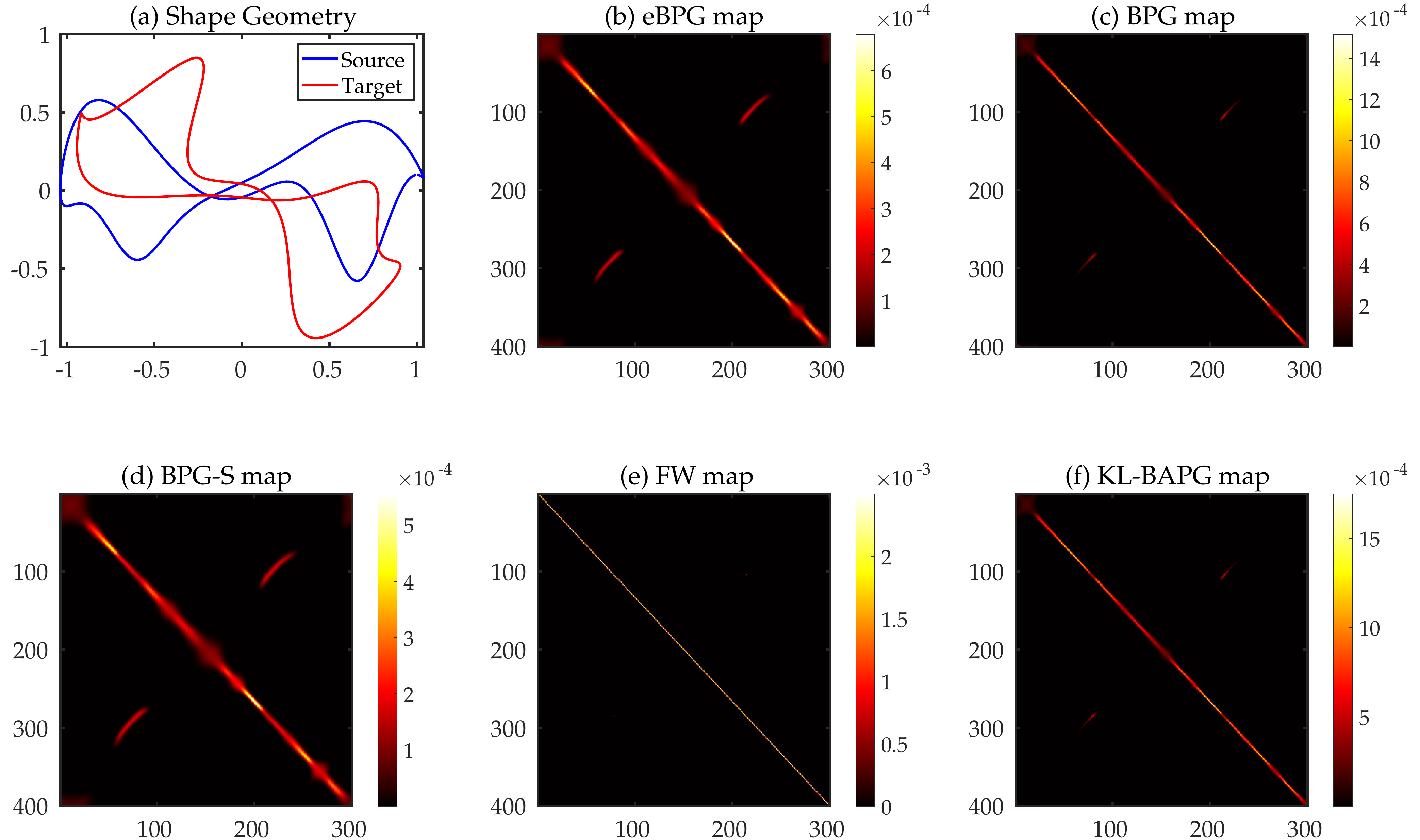}
	\caption{(a): 2D shape geometry of the source and target; (b)-(f): visualization of coupling matrix. }
	\vspace{-3mm}
\label{fig:syn}
\end{figure*}

 Figs.\ref{fig:syn} (b)-(f) provide all color maps of coupling matrices to visualize the matching results. Here, the sparser coupling matrices indicate sharper mapping. All experiment results are consistent with those in Table \ref{tab:algo}. We can observe that both BPG and FW give us satisfactory solution performance as they aim at solving the GW problem exactly. However, FW and BPG will suffer from a significant computation burden which will be further justified in Sec \ref{sec:graph_align} and \ref{sec:graph_partition}. On the other hand, the performance of BPG-S and eBPG is obviously harmed by the inexacness issue. The sharpness  of KL-BAPG's coupling matrix is relatively unaffected by its infeasibility issue too much, although its coupling matrix is denser than BPG and FW ones. 
 As we shall see later,  the effect of the infeasibility issue is minor when the penalty parameter $\rho$ is not too small and will not  even result in  any real cost for graph alignment and partition tasks, which only care about the matching accuracy instead of the sharpness of the coupling.

\subsection{Graph Alignment}\label{sec:graph_align} 
Graph alignment aims to identify the node correspondence between two graphs possibly with different topology structures~\citep{zhang2021balancing,chen2020consistent}. Instead of solving the restricted quadratic assignment problem~\citep{lawler1963quadratic,lacoste2006word}, the GW distance provides an optimal probabilistic correspondence relationship via preservation of the isometric property. 
Here, we compare the proposed KL-BAPG with all existing baselines: FW~\citep{titouan2019optimal}, BPG~\citep{xu2019gromov}, BPG-S~\citep{xu2019gromov} (i.e., the only difference between BPG and BPG-S is that the number of inner iterations for BPG-S is just one), ScalaGW~\citep{xu2019scalable}, SpecGW~\citep{chowdhury2021generalized}, and eBPG~\citep{solomon2016entropic}. Except for BPG and eBPG, others are pure heuristic methods without any theoretical guarantee. Besides the GW-based methods, we also consider three widely used non-GW graph alignment baselines, including IPFP \citep{ipfp}, RRWM \citep{rrwm}, and SpecMethod \citep{sm}.

\paragraph{Parameters Setup} We utilize the unweighted symmetric adjacency matrices as our input distance matrices, i.e., $D_X$ and  $D_Y$. Alternatively,  SpecGW uses the heat kernel $\exp(-L)$ where $L$ is the normalized graph Laplacian matrix. We set both $\mu$ and $\nu$ to be the uniform distribution.  For three heuristic methods --- BPG-S, ScalaGW, and SpecGW, we follow the same setup reported in their papers. As mentioned, eBPG is very sensitive to the entropic regularization parameter. To get comparable results, we report the best result among the set $\{0.1, 0.01, 0.001\}$ of the regularization parameter. For BPG and KL-BAPG, we use the constant step size $5$ and $\rho=0.1$ respectively.  For FW, we use the default implementation in the PythonOT package (see Appendix \ref{sec:exp}). All the experiment results reported here were the average of  5 independent trials over different random seeds and the standard deviation is collected in Appendix \ref{sec:exp}. 

 \begin{table*}[!t]
\caption{Comparison of the matching accuracy (\%) and wall-clock time (seconds) on graph alignment. For KL-BAPG, we also report the time of GPU implementation.}
\centering
\small
\resizebox{\linewidth}{!}{
\begin{tabular}{c|cr|ccr|ccr|ccr}
\toprule

\multirow{2}{*}{Method} &\multicolumn{2}{c|}{Synthetic} & \multicolumn{3}{c|}{Proteins}  & \multicolumn{3}{c|}{Enzymes} & \multicolumn{3}{c}{Reddit} \\
            & Acc & Time & Raw & Noisy & Time & Raw & Noisy & Time & Raw & Noisy & Time  \\\midrule
IPFP & - & - & 43.84 & 29.89 & 87.0 & 40.37 & 27.39 & 23.7 & - & - & -\\
RRWM & - & - & 71.79 & 33.92 & 239.3 & 60.56 & 30.51 & 114.1 & - & - & -\\
SpecMethod   & - & - & 72.40 & 22.92 & 40.5 & 71.43 & 21.39 & 9.6 & - & - & -\\\midrule
FW & 24.50 & 8182 & 29.96 & 20.24 & 54.2 & 32.17 & 22.80 & 10.8 & 21.51 & 17.17 & 1121\\
ScalaGW & 17.93& 12002& 16.37 & 16.05 &372.2  & 12.72 & 11.46 &213.0&  0.54 & 0.70 &1109\\
SpecGW & 13.27&  1462& 78.11 & 19.31 &\textbf{30.7}  & 79.07 & 21.14 & \textbf{6.7}& 50.71 &19.66 &1074\\
eBPG & 34.33&  9502& 67.48 & 45.85 &208.2 & 78.25 & 60.46 &499.7&  3.76 & 3.34 &1234\\
BPG         & 57.56& 22600 & 71.99 & 52.46 &130.4 & 79.19 & 62.32 &73.1& 39.04 &36.68 &1907\\
BPG-S & 61.48& 18587& 71.74 & 52.74 &40.4  & 79.25 & 62.21 &13.4& 39.04 &36.68 &1431\\\midrule
{KL-BAPG}        & \textbf{99.79}&  9024& \textbf{78.18} & \textbf{57.16} &59.1  & \textbf{79.66} & \textbf{62.85} &6.9& \textbf{50.93} &\textbf{49.45} & 780 \\
{KL-BAPG-GPU}    & - & \textbf{1253} & - & - & 75.4 & - & - & 21.8 & - & - &\textbf{115}\\\bottomrule
\end{tabular}
}
\label{tab:align_result} 
\vspace{-6mm}
\end{table*}

\vspace{-3mm}
\paragraph{Database Statistics} We test all methods on both synthetic and real-world databases. Our setup for the synthetic database is the same as in \citep{xu2019gromov}. The source graph $\mathcal G_s = \{\mathcal V_s, \mathcal E_s\}$ is generated by two ideal random models, Gaussian random partition and Barabasi-Albert models,  with different scales, i.e., $|\mathcal V_s| \in \{500, 1000, 1500, 2000, 2500\}$. Then, we generate the target graph $\mathcal G_t=\{\mathcal V_t, \mathcal E_t\}$ by first adding $q\%$ noisy nodes to the source graph, and then generating $q\%$ noisy edges between the nodes in $\mathcal V_t$, i.e., $|\mathcal V_t|=(1+q\%)|\mathcal V_s|, |\mathcal E_t|=(1+q\%)|\mathcal E_s|$, where $q\in \{0, 10, 20, 30, 40, 50\}$. For each setup, we generate five synthetic graph pairs over different random seeds. To sum up, the synthetic database contains 300 different graph pairs. We also validate our proposed methods on three other real-world databases from \citep{chowdhury2021generalized}, including two biological graph databases \textit{Proteins} and \textit{Enzymes}, and a social network database \textit{Reddit}. Furthermore, to demonstrate the robustness of our method regarding the noise level, we follow the noise-generating process (i.e., $q = 10\%$) conducted for the synthesis case to create new databases on top of the three real-world databases. Toward that end, the statistics of all databases used for the graph alignment task have been summarized in Appendix \ref{sec:exp}. We match each node in $\mathcal G_s$ with the most likely node in $\mathcal G_t$ according to the optimized $\pi^\star$. Given the predicted correspondence set $\mathcal S_{\text{pred}}$ and the ground-truth correspondence set $\mathcal S_{\text{gt}}$, we calculate the matching accuracy by $\textnormal{Acc}=|\mathcal S_{\text{gt}} \cap \mathcal S_{\text{pred}}|/|\mathcal S_{\text{gt}}|\times 100\%$.


\paragraph{Results of Our Methods} Table \ref{tab:align_result} shows the
comparison of matching accuracy and wall-clock time on four databases. We observe that {KL-BAPG} works exceptionally well both in terms of computational time and accuracy, 
especially for two large-scale noisy graph databases \textit{Synthetic} and \textit{Reddit}. Notably, {KL-BAPG} is robust enough so that it is not necessary to perform parameter tuning.
As we mentioned in Sec \ref{sec:method}, the effectiveness of GPU acceleration for {KL-BAPG} is also well corroborated on \textit{Synthetic} and \textit{Reddit}. GPU cannot further speed up the training time of \textit{Proteins} and \textit{Reddit} as graphs in these two databases are small-scale. 
Additional experiment results to demonstrate the robustness of {KL-BAPG} and its GPU acceleration will be given in Sec \ref{sec:sentiv}.


\paragraph{Comparison with Other Methods}
Traditional non-GW graph alignment methods (IPFP, RRWM, and SpecMethod) have the out-of-memory issue on graphs with more than 500 nodes (e.g., Synthetic and Reddit) and are sensitive to noise. The performance of eBPG  and ScalaGW is influenced by the entropic regularization parameter and approximation error respectively, which accounts for their poor performance. Moreover, it is easy to observe that SpecGW works pretty well on the small dataset but the performance degrades dramatically on the large one, e.g., \textit{synthetic}.  The reason is that SpecGW relies on a linear programming solver as its subroutine, which is not well-suited for large-scale settings. Besides, although ScalaGW has the lowest per-iteration computational complexity, the recursive $K$-partition mechanism developed in \citep{xu2019scalable} is not friendly to parallel computing. Therefore, ScalaGW does not demonstrate attractive performance on multi-core processors. 
\subsection{Graph Partition}\label{sec:graph_partition}
The GW distance can also be potentially applied to the graph partition task. That is, we are trying to match the source graph with a disconnected target graph having $K$ isolated and self-connected super nodes, where $K$ is the number of clusters~\citep{abrishami2020geometry,li2020dirichlet}. Similarly, we compare the proposed KL-BAPG with the other baselines described in Sec \ref{sec:graph_align} on four real-world graph partitioning datasets. Following \citep{chowdhury2021generalized}, we also add three non-GW methods specialized in graph alignment, including FastGreedy \citep{clauset2004finding}, Louvain \citep{blondel2008fast}, and Infomap \citep{rosvall2008maps}.

\vspace{-3mm}
\paragraph{Parameters Setup} For the input distance matrices $D_X$ and $D_Y$, we test our methods on both the adjacency matrices and the heat kernel matrices proposed in \citep{chowdhury2021generalized}. For {KL-BAPG}, we pick the lowest converged function value among $\rho \in\{0.1,0.05,0.01\}$ for adjacency matrices and $\rho \in\{0.001,0.0005,0.0001\}$ for heat kernel matrices. The quality of graph partition results is quantified by computing the adjusted mutual information (AMI) score \citep{vinh2010information} against the ground-truth partition. 
\begin{table*}[]
\centering
\caption{Comparison of AMI scores on graph partition datasets using the adjacency matrices and the heat kernel matrices.}
\resizebox{0.90\linewidth}{!}{
\begin{tabular}{cc|cccccccc}
\toprule
\multirow{2}{*}{Category} & \multirow{2}{*}{Method} & \multicolumn{2}{c}{Wikipedia}   & \multicolumn{2}{c}{EU-email}    & \multicolumn{2}{c}{Amazon}      & \multicolumn{2}{c}{Village}     \\

                           &                          & Raw            & Noisy          & Raw            & Noisy          & Raw            & Noisy          & Raw            & Noisy          \\ \midrule
\multirow{3}{*}{Non-GW}& 
FastGreedy& 0.382&0.341&0.312&0.251&0.637&0.573&\textbf{0.881}&0.778\\
& Louvain   & 0.377 & 0.329 & 0.447 & 0.382 & 0.622 & \textbf{0.584} & \textbf{0.881} & \textbf{0.827}\\
&  Infomap   & 0.332 & 0.329 & 0.374 & 0.379 & \textbf{0.940} & 0.463 & \textbf{0.881} & 0.190\\\midrule
\multirow{5}{*}{Adjacency}  & FW &0.341 & 0.323 & 0.440 & 0.409 & 0.374 & 0.338 & 0.684 & 0.539\\
                           & eBPG                     & 0.461          & \textbf{0.413} & \textbf{0.517} & 0.422 & 0.429          & \textbf{0.387}          & 0.703          & 0.658          \\
                           & BPG                      & 0.367          & 0.333          & 0.478          & 0.414          & 0.412          & 0.368          & 0.642          & 0.575          \\
                           & BPG-S                      & 0.357          & 0.285          & 0.451          & 0.404          & 0.443         & 0.352          & 0.606          & 0.560          \\
    & {KL-BAPG}                     & \textbf{0.469}   & 0.396          & 0.508          & \textbf{0.428}          & \textbf{0.457}          & 0.362  & \textbf{0.736} & \textbf{0.681} \\ \midrule
\multirow{5}{*}{Heat Kernel}  & SpecGW                  & 0.442          &\textbf{0.395}          & 0.487          & 0.425          & 0.565          & 0.487          & 0.758          & 0.707  \\
                           & {eBPG}                     & {0.100}          & {0.082}          & {0.011}          & {0.188}          & {0.604}          & {0.031}          & {0.002}          & {0.003}          \\
                           & BPG                      & 0.418          & 0.373          & 0.473          & 0.253          & 0.492          & 0.436          & 0.705          & 0.619          \\
                            & BPG-S                      & 0.411          & 0.373          & 0.475          & 0.253          & 0.483          & 0.425          & 0.642          & 0.619          \\
    & {KL-BAPG}                     & \textbf{0.533}   & 0.365 & \textbf{0.533} & \textbf{0.436} & \textbf{0.630} & \textbf{0.502} & \textbf{0.797} & \textbf{0.711} \\ \bottomrule
\end{tabular}
}
\vspace{-4mm}
\label{tab:partition_result} 
\end{table*}

\vspace{-3mm}
\paragraph{Results of All Methods} Table \ref{tab:partition_result} shows the comparison of AMI scores among all methods for graph partition. {KL-BAPG} outperforms other GW-based methods in most cases and is more robust under noisy conditions. Specifically, KL-BAPG is consistently better than both FW and SpecGW, which  rely on the Frank-Wolfe method to solve the problem. eBPG has comparable results when using the adjacency matrices, but is sensitive to process the spectral matrices.
The possible reason is that the adjacency matrix and the heat kernel matrix have quite different structures, e.g., the former is sparse while the latter is dense. BPG and BPG-S enjoy similar performances in most cases, but they are not as good as our proposed {KL-BAPG} on all datasets. KL-BAPG also shows competitive performance compared to specialized non-GW graph partition methods. For example, {KL-BAPG}  outperforms Infomap and Louvain in 6 and 4 datasets out of 8, respectively



\subsection{The Effectiveness and Robustness of KL-BAPG}
\label{sec:sentiv}
At first, we target at demonstrating the robustness of KL-BAPG  on graph alignment, as it is more reasonable to test the robustness of a method on a database (e.g., graph alignment) rather than a single point (e.g., graph partition).

\paragraph{Noise Level and Graph Scale} At the beginning, we present the sensitivity analysis of KL-BAPG with respect to the noise level $q\%$ and the graph scale $|\mathcal V_s|$ in Fig. \ref{fig:noise} using the Synthetic Database. Surprisingly, the solution performance of KL-BAPG is robust to both the noise level and graph scale. In contrast, the accuracy of other methods degrades dramatically as the noise level or the graph scale increases. 
\vspace{-3mm}
\paragraph{Trade-off among Efficiency, Accuracy and Feasibility}
We present a unified perspective on the trade-off between efficiency, accuracy, and feasibility for all GW-based algorithms on the Reddit database in Fig. \ref{fig:noise} (b). As shown, our proposed KL-BAPG is able to achieve a desirable balance between these three factors.
Table \ref{tab:mar_compare} provides a detailed comparison of the four databases, while Table \ref{tab:mar_compare2} demonstrates the robustness of KL-BAPG with respect to the step size $\rho$. This experiment supports the validity of Proposition \ref{prop:bapg_approx} and provides practical guidance on choosing the optimal step size. Note that a larger $\rho$ leads to a lower infeasibility error but a slower convergence rate.

\begin{table}[t!]
\centering
\label{tab:gpu}
\caption{GPU \& CPU wall-clock time comparison of KL-BAPG, BPG, and eBPG on graph alignment.}
\resizebox{0.8\linewidth}{!}{

\begin{tabular}{c|ccc|ccc}
\toprule
                    & \multicolumn{3}{c|}{Reddit Dataset} & \multicolumn{3}{c}{Synthetic Dataset} \\
                    & KL-BAPG            & BPG    & eBPG   & KL-BAPG             & BPG      & eBPG    \\\midrule
CPU Time(s) & 780              & 1907   & 1234   & 9024             & 22600    & 9502    \\
GPU  Time(s) & 115              & 1013   & 2274   & 1253             & 4458     & 2709    \\
Acceleration Ratio  & \textbf{6.78}    & 1.88   & 0.54   & \textbf{7.20}    & 5.07     & 3.51   \\\bottomrule
\end{tabular}
}
\label{tab:gpu}
\end{table}

\begin{figure*}[t!]
	\vspace{-4mm}
	\centering
	\subfigure[]{\includegraphics[width=0.6\textwidth]{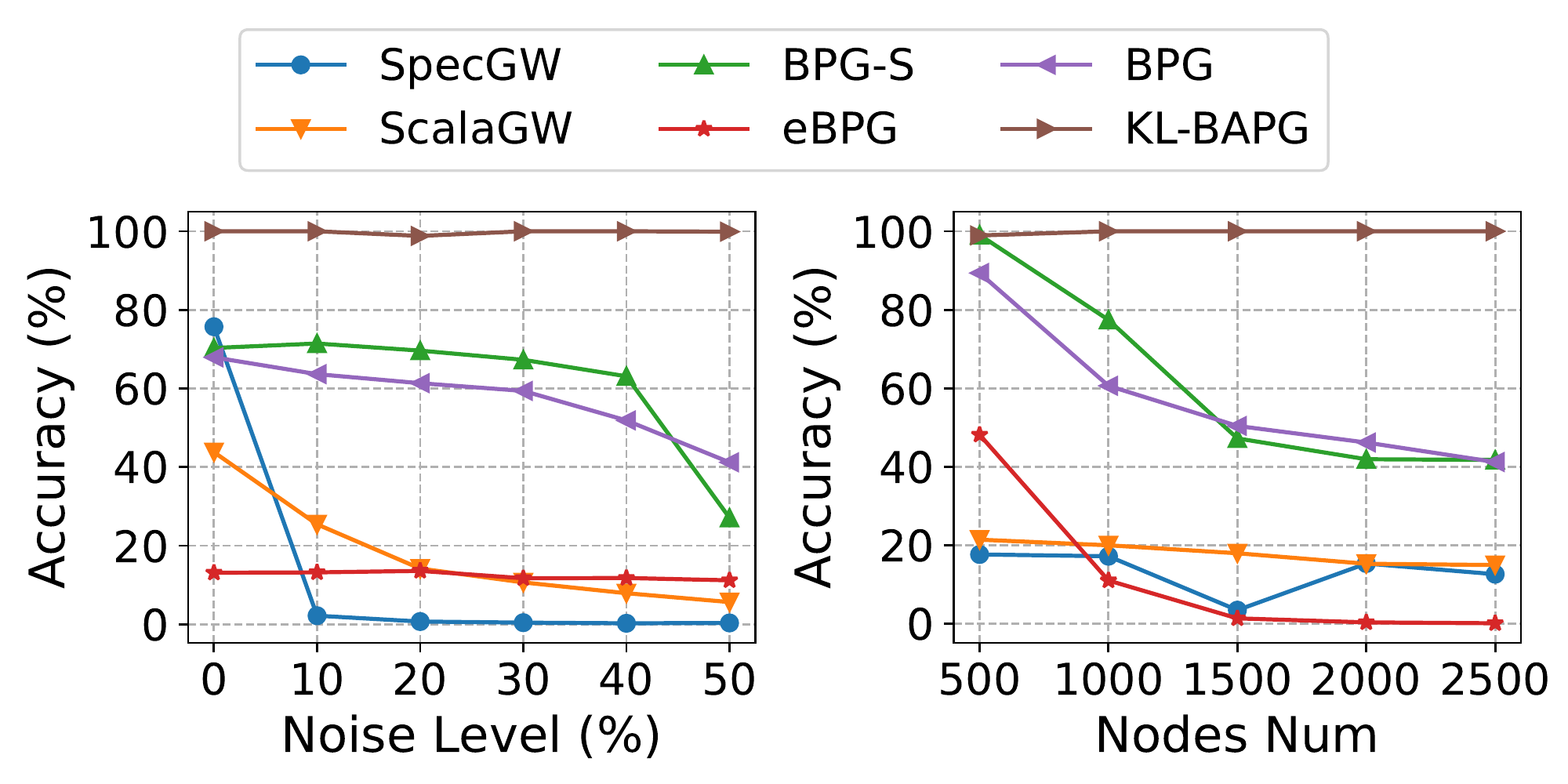}} 
	\subfigure[]{\includegraphics[width=0.39\textwidth]{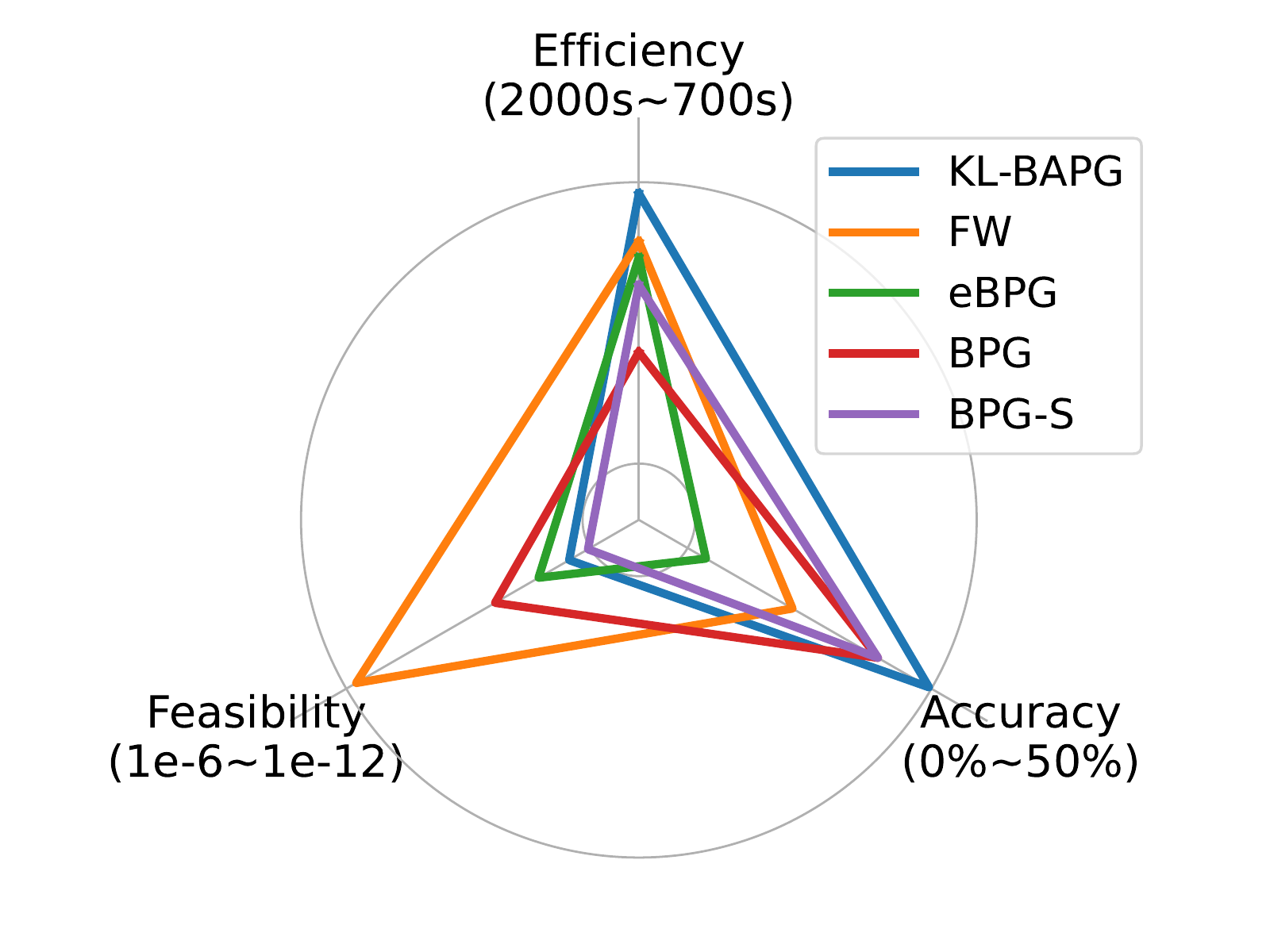}} 
\caption{(a) Sensitivity of the noise level and graph scale on the synthetic graph alignment database. (b) Visualization of trade-off among efficiency, accuracy and feasibility on the Reddit database. The infeasibility error is computed by  $\|\pi^T\mathbf{1}_n -\nu\|+\|\pi \mathbf{1}_m -\mu\|$. "Closer to the boundary of the outer cycle generally indicates higher accuracy, faster speed, and lower infeasibility error.}\label{fig:noise}
\end{figure*}
\paragraph{GPU Acceleration of KL-BAPG.} We conduct experiment results to further justify that KL-BAPG is GPU-friendly. In Table \ref{tab:gpu}, We compare the acceleration ratio (i.e., CPU wall-clock time divided by GPU wall-clock time) of KL-BAPG, eBPG, and BPG on two large-scale graph alignment datasets using the same computing server. For eBPG, we use the official implementation in the PythonOT package, which supports running on GPU. For BPG, we implement the GPU version by ourselves using Pytorch. We can find that KL-BAPG has a much higher acceleration ratio on the GPU compared to BPG and eBPG.
\vspace{-4mm}
\section{Closing Remark}
\vspace{-3mm}
In this study, we have explored the development of a single-loop algorithm for the relaxation of GW distance computation. By utilizing an error bound condition that was not previously investigated in the GW literature, we have successfully conducted the convergence analysis of BAPG. However, the proposed algorithm still faces the challenge of cubic per-iteration computational complexity, which limits its applicability to large-scale real-world problems. A potential future direction is to incorporate sparse and low-rank structures in the matching matrix to reduce the per-iteration cost and improve the performance of the algorithm. Additionally, our method can also be applied to non-symmetric distance matrices, as the Luo-tseng error bound condition remains valid. On another note, our work also provides a new perspective for the study of the alternating projected descent method for general non-convex problems, which is an open area of research in the optimization community.
\section*{Acknowledgement}
 Material in this paper is based upon work supported by the Air Force Office
of Scientific Research under award number FA9550-20-1-0397. Additional support is gratefully
acknowledged from NSF grants 1915967 and 2118199. Jianheng Tang and Jia Li are supported by NSFC Grant 62206067 and Guangzhou-HKUST(GZ) Joint Funding Scheme. Anthony Man-Cho So is supported by in part by the Hong Kong RGC GRF project CUHK 14203920.
\bibliography{iclr2023_conference}

\begin{thebibliography}{53}
\providecommand{\natexlab}[1]{#1}
\providecommand{\url}[1]{\texttt{#1}}
\expandafter\ifx\csname urlstyle\endcsname\relax
  \providecommand{\doi}[1]{doi: #1}\else
  \providecommand{\doi}{doi: \begingroup \urlstyle{rm}\Url}\fi

\bibitem[Abrishami et~al.(2020)Abrishami, Guillen, Rule, Schutzman, Solomon,
  Weighill, and Wu]{abrishami2020geometry}
Tara Abrishami, Nestor Guillen, Parker Rule, Zachary Schutzman, Justin Solomon,
  Thomas Weighill, and Si~Wu.
\newblock Geometry of graph partitions via optimal transport.
\newblock \emph{SIAM Journal on Scientific Computing}, 42\penalty0
  (5):\penalty0 A3340--A3366, 2020.

\bibitem[Attouch et~al.(2010)Attouch, Bolte, Redont, and
  Soubeyran]{attouch2010proximal}
H{\'e}dy Attouch, J{\'e}r{\^o}me Bolte, Patrick Redont, and Antoine Soubeyran.
\newblock Proximal alternating minimization and projection methods for
  nonconvex problems: An approach based on the kurdyka-{\l}ojasiewicz
  inequality.
\newblock \emph{Mathematics of operations research}, 35\penalty0 (2):\penalty0
  438--457, 2010.

\bibitem[Attouch et~al.(2013)Attouch, Bolte, and
  Svaiter]{attouch2013convergence}
Hedy Attouch, J{\'e}r{\^o}me Bolte, and Benar~Fux Svaiter.
\newblock Convergence of descent methods for semi-algebraic and tame problems:
  proximal algorithms, forward--backward splitting, and regularized
  gauss--seidel methods.
\newblock \emph{Mathematical Programming}, 137\penalty0 (1):\penalty0 91--129,
  2013.

\bibitem[Bauschke(1996)]{bauschke1996projection}
Heinz~H. Bauschke.
\newblock \emph{Projection Algorithms and Monotone Operators}.
\newblock PhD thesis, Simon Fraser University, 1996.

\bibitem[Blondel et~al.(2008)Blondel, Guillaume, Lambiotte, and
  Lefebvre]{blondel2008fast}
Vincent~D Blondel, Jean-Loup Guillaume, Renaud Lambiotte, and Etienne Lefebvre.
\newblock Fast unfolding of communities in large networks.
\newblock \emph{Journal of statistical mechanics: theory and experiment},
  2008\penalty0 (10):\penalty0 P10008, 2008.

\bibitem[Bunne et~al.(2019)Bunne, Alvarez-Melis, Krause, and
  Jegelka]{bunne2019learning}
Charlotte Bunne, David Alvarez-Melis, Andreas Krause, and Stefanie Jegelka.
\newblock Learning generative models across incomparable spaces.
\newblock In \emph{International Conference on Machine Learning}, pp.\
  851--861. PMLR, 2019.

\bibitem[Cao et~al.(2020)Cao, Bai, Hong, and Wan]{cao2020unsupervised}
Kai Cao, Xiangqi Bai, Yiguang Hong, and Lin Wan.
\newblock Unsupervised topological alignment for single-cell multi-omics
  integration.
\newblock \emph{Bioinformatics}, 36\penalty0 (Supplement\_1):\penalty0
  i48--i56, 2020.

\bibitem[Cao et~al.(2022)Cao, Hong, and Wan]{cao2022manifold}
Kai Cao, Yiguang Hong, and Lin Wan.
\newblock Manifold alignment for heterogeneous single-cell multi-omics data
  integration using pamona.
\newblock \emph{Bioinformatics}, 38\penalty0 (1):\penalty0 211--219, 2022.

\bibitem[Chen et~al.(2020)Chen, Heimann, Vahedian, and
  Koutra]{chen2020consistent}
Xiyuan Chen, Mark Heimann, Fatemeh Vahedian, and Danai Koutra.
\newblock Consistent network alignment via proximity-preserving node embedding.
\newblock \emph{arXiv preprint arXiv:2005.04725}, 2020.

\bibitem[Cho et~al.(2010)Cho, Lee, and Lee]{rrwm}
Minsu Cho, Jungmin Lee, and Kyoung~Mu Lee.
\newblock Reweighted random walks for graph matching.
\newblock In \emph{European conference on Computer vision}, pp.\  492--505.
  Springer, 2010.

\bibitem[Chowdhury \& M{\'e}moli(2019)Chowdhury and
  M{\'e}moli]{chowdhury2019gromov}
Samir Chowdhury and Facundo M{\'e}moli.
\newblock The gromov--wasserstein distance between networks and stable network
  invariants.
\newblock \emph{Information and Inference: A Journal of the IMA}, 8\penalty0
  (4):\penalty0 757--787, 2019.

\bibitem[Chowdhury \& Needham(2021)Chowdhury and
  Needham]{chowdhury2021generalized}
Samir Chowdhury and Tom Needham.
\newblock Generalized spectral clustering via gromov-wasserstein learning.
\newblock In \emph{International Conference on Artificial Intelligence and
  Statistics}, pp.\  712--720. PMLR, 2021.

\bibitem[Clauset et~al.(2004)Clauset, Newman, and Moore]{clauset2004finding}
Aaron Clauset, Mark~EJ Newman, and Cristopher Moore.
\newblock Finding community structure in very large networks.
\newblock \emph{Physical review E}, 70\penalty0 (6):\penalty0 066111, 2004.

\bibitem[Cuturi(2013)]{cuturi2013sinkhorn}
Marco Cuturi.
\newblock Sinkhorn distances: Lightspeed computation of optimal transport.
\newblock \emph{Advances in neural information processing systems},
  26:\penalty0 2292--2300, 2013.

\bibitem[Flamary et~al.(2021)Flamary, Courty, Gramfort, Alaya, Boisbunon,
  Chambon, Chapel, Corenflos, Fatras, Fournier, Gautheron, Gayraud, Janati,
  Rakotomamonjy, Redko, Rolet, Schutz, Seguy, Sutherland, Tavenard, Tong, and
  Vayer]{flamary2021pot}
R{\'e}mi Flamary, Nicolas Courty, Alexandre Gramfort, Mokhtar~Z. Alaya,
  Aur{\'e}lie Boisbunon, Stanislas Chambon, Laetitia Chapel, Adrien Corenflos,
  Kilian Fatras, Nemo Fournier, L{\'e}o Gautheron, Nathalie~T.H. Gayraud,
  Hicham Janati, Alain Rakotomamonjy, Ievgen Redko, Antoine Rolet, Antony
  Schutz, Vivien Seguy, Danica~J. Sutherland, Romain Tavenard, Alexander Tong,
  and Titouan Vayer.
\newblock Pot: Python optimal transport.
\newblock \emph{Journal of Machine Learning Research}, 22\penalty0
  (78):\penalty0 1--8, 2021.
\newblock URL \url{http://jmlr.org/papers/v22/20-451.html}.

\bibitem[Gao et~al.(2021)Gao, Huang, and Li]{gao2021unsupervised}
Ji~Gao, Xiao Huang, and Jundong Li.
\newblock Unsupervised graph alignment with wasserstein distance discriminator.
\newblock In \emph{Proceedings of the 27th ACM SIGKDD Conference on Knowledge
  Discovery \& Data Mining}, pp.\  426--435, 2021.

\bibitem[Krichene et~al.(2015)Krichene, Krichene, and
  Bayen]{krichene2015efficient}
Walid Krichene, Syrine Krichene, and Alexandre Bayen.
\newblock Efficient bregman projections onto the simplex.
\newblock In \emph{2015 54th IEEE Conference on Decision and Control (CDC)},
  pp.\  3291--3298. IEEE, 2015.

\bibitem[Lacoste-Julien et~al.(2006)Lacoste-Julien, Taskar, Klein, and
  Jordan]{lacoste2006word}
Simon Lacoste-Julien, Ben Taskar, Dan Klein, and Michael Jordan.
\newblock Word alignment via quadratic assignment.
\newblock 2006.

\bibitem[Lawler(1963)]{lawler1963quadratic}
Eugene~L Lawler.
\newblock The quadratic assignment problem.
\newblock \emph{Management science}, 9\penalty0 (4):\penalty0 586--599, 1963.

\bibitem[Leordeanu \& Hebert(2005)Leordeanu and Hebert]{sm}
Marius Leordeanu and Martial Hebert.
\newblock A spectral technique for correspondence problems using pairwise
  constraints.
\newblock In \emph{International Conference on Computer Vision}, pp.\
  1482--1489. IEEE, 2005.

\bibitem[Leordeanu et~al.(2009)Leordeanu, Hebert, and Sukthankar]{ipfp}
Marius Leordeanu, Martial Hebert, and Rahul Sukthankar.
\newblock An integer projected fixed point method for graph matching and map
  inference.
\newblock \emph{Advances in neural information processing systems}, 22, 2009.

\bibitem[Li et~al.(2020)Li, Yu, Li, Zhang, Zhao, Rong, Cheng, and
  Huang]{li2020dirichlet}
Jia Li, Jianwei Yu, Jiajin Li, Honglei Zhang, Kangfei Zhao, Yu~Rong, Hong
  Cheng, and Junzhou Huang.
\newblock Dirichlet graph variational autoencoder.
\newblock \emph{Advances in Neural Information Processing Systems},
  33:\penalty0 5274--5283, 2020.

\bibitem[Li et~al.(2021)Li, Liu, Zhang, Wang, Wen, Pan, and Cheng]{li2021mask}
Jia Li, Mengzhou Liu, Honglei Zhang, Pengyun Wang, Yong Wen, Lujia Pan, and
  Hong Cheng.
\newblock Mask-gvae: Blind denoising graphs via partition.
\newblock In \emph{Proceedings of the Web Conference 2021}, pp.\  3688--3698,
  2021.

\bibitem[Luo \& Tseng(1992)Luo and Tseng]{luo1992error}
Zhi-Quan Luo and Paul Tseng.
\newblock Error bound and convergence analysis of matrix splitting algorithms
  for the affine variational inequality problem.
\newblock \emph{SIAM Journal on Optimization}, 2\penalty0 (1):\penalty0 43--54,
  1992.

\bibitem[Mai et~al.(2021)Mai, Lindb{\"a}ck, and Johansson]{mai2021fast}
Vien~V Mai, Jacob Lindb{\"a}ck, and Mikael Johansson.
\newblock A fast and accurate splitting method for optimal transport: Analysis
  and implementation.
\newblock \emph{arXiv preprint arXiv:2110.11738}, 2021.

\bibitem[M{\'e}moli(2009)]{memoli2009spectral}
Facundo M{\'e}moli.
\newblock Spectral gromov-wasserstein distances for shape matching.
\newblock In \emph{2009 IEEE 12th International Conference on Computer Vision
  Workshops, ICCV Workshops}, pp.\  256--263. IEEE, 2009.

\bibitem[M{\'e}moli(2011)]{memoli2011gromov}
Facundo M{\'e}moli.
\newblock Gromov--wasserstein distances and the metric approach to object
  matching.
\newblock \emph{Foundations of computational mathematics}, 11\penalty0
  (4):\penalty0 417--487, 2011.

\bibitem[M{\'e}moli(2014)]{memoli2014gromov}
Facundo M{\'e}moli.
\newblock The gromov--wasserstein distance: A brief overview.
\newblock \emph{Axioms}, 3\penalty0 (3):\penalty0 335--341, 2014.

\bibitem[M{\'e}moli \& Sapiro(2004)M{\'e}moli and Sapiro]{memoli2004comparing}
Facundo M{\'e}moli and Guillermo Sapiro.
\newblock Comparing point clouds.
\newblock In \emph{Proceedings of the 2004 Eurographics/ACM SIGGRAPH symposium
  on Geometry processing}, pp.\  32--40, 2004.

\bibitem[Nedi{\'c}(2011)]{nedic2011random}
Angelia Nedi{\'c}.
\newblock Random algorithms for convex minimization problems.
\newblock \emph{Mathematical programming}, 129\penalty0 (2):\penalty0 225--253,
  2011.

\bibitem[Nocedal \& Wright(2006)Nocedal and Wright]{nocedal2006numerical}
Jorge Nocedal and Stephen Wright.
\newblock \emph{Numerical optimization}.
\newblock Springer Science \& Business Media, 2006.

\bibitem[Peyr{\'e} et~al.(2016)Peyr{\'e}, Cuturi, and Solomon]{peyre2016gromov}
Gabriel Peyr{\'e}, Marco Cuturi, and Justin Solomon.
\newblock Gromov-wasserstein averaging of kernel and distance matrices.
\newblock In \emph{International Conference on Machine Learning}, pp.\
  2664--2672. PMLR, 2016.

\bibitem[Peyr{\'e} et~al.(2019)Peyr{\'e}, Cuturi,
  et~al.]{peyre2019computational}
Gabriel Peyr{\'e}, Marco Cuturi, et~al.
\newblock Computational optimal transport: With applications to data science.
\newblock \emph{Foundations and Trends{\textregistered} in Machine Learning},
  11\penalty0 (5-6):\penalty0 355--607, 2019.

\bibitem[Rosvall \& Bergstrom(2008)Rosvall and Bergstrom]{rosvall2008maps}
Martin Rosvall and Carl~T Bergstrom.
\newblock Maps of random walks on complex networks reveal community structure.
\newblock \emph{Proceedings of the national academy of sciences}, 105\penalty0
  (4):\penalty0 1118--1123, 2008.

\bibitem[S{\'e}journ{\'e} et~al.(2021)S{\'e}journ{\'e}, Vialard, and
  Peyr{\'e}]{sejourne2021unbalanced}
Thibault S{\'e}journ{\'e}, Fran{\c{c}}ois-Xavier Vialard, and Gabriel
  Peyr{\'e}.
\newblock The unbalanced gromov wasserstein distance: Conic formulation and
  relaxation.
\newblock \emph{Advances in Neural Information Processing Systems},
  34:\penalty0 8766--8779, 2021.

\bibitem[Solomon et~al.(2016)Solomon, Peyr{\'e}, Kim, and
  Sra]{solomon2016entropic}
Justin Solomon, Gabriel Peyr{\'e}, Vladimir~G Kim, and Suvrit Sra.
\newblock Entropic metric alignment for correspondence problems.
\newblock \emph{ACM Transactions on Graphics (TOG)}, 35\penalty0 (4):\penalty0
  1--13, 2016.

\bibitem[Tang et~al.(2023)Tang, Zhang, Li, Zhao, Tsung, and Li]{tang2023robust}
Jianheng Tang, Weiqi Zhang, Jiajin Li, Kangfei Zhao, Fugee Tsung, and Jia Li.
\newblock Robust attributed graph alignment via joint structure learning and
  optimal transport.
\newblock \emph{ICDE}, 2023.

\bibitem[Teichmann \& Efremova(2020)Teichmann and
  Efremova]{teichmann2020method}
Sarah Teichmann and Mirjana Efremova.
\newblock Method of the year 2019: single-cell multimodal omics.
\newblock \emph{Nature Methods}, 17\penalty0 (1):\penalty0 2020, 2020.

\bibitem[Vayer et~al.(2018)Vayer, Chapel, Flamary, Tavenard, and
  Courty]{vayer2018fused}
Titouan Vayer, Laetita Chapel, R{\'e}mi Flamary, Romain Tavenard, and Nicolas
  Courty.
\newblock Fused gromov-wasserstein distance for structured objects: theoretical
  foundations and mathematical properties.
\newblock \emph{arXiv preprint arXiv:1811.02834}, 2018.

\bibitem[Vayer et~al.(2019{\natexlab{a}})Vayer, Courty, Tavenard, and
  Flamary]{titouan2019optimal}
Titouan Vayer, Nicolas Courty, Romain Tavenard, and R{\'e}mi Flamary.
\newblock Optimal transport for structured data with application on graphs.
\newblock In \emph{International Conference on Machine Learning}, pp.\
  6275--6284. PMLR, 2019{\natexlab{a}}.

\bibitem[Vayer et~al.(2019{\natexlab{b}})Vayer, Flamary, Tavenard, Chapel, and
  Courty]{vayer2019sliced}
Titouan Vayer, R{\'e}mi Flamary, Romain Tavenard, Laetitia Chapel, and Nicolas
  Courty.
\newblock Sliced gromov-wasserstein.
\newblock In \emph{NeurIPS 2019-Thirty-third Conference on Neural Information
  Processing Systems}, volume~32, 2019{\natexlab{b}}.

\bibitem[Vincent-Cuaz et~al.(2021{\natexlab{a}})Vincent-Cuaz, Flamary, Corneli,
  Vayer, and Courty]{vincent2021semi}
C{\'e}dric Vincent-Cuaz, R{\'e}mi Flamary, Marco Corneli, Titouan Vayer, and
  Nicolas Courty.
\newblock Semi-relaxed gromov wasserstein divergence with applications on
  graphs.
\newblock \emph{arXiv preprint arXiv:2110.02753}, 2021{\natexlab{a}}.

\bibitem[Vincent-Cuaz et~al.(2021{\natexlab{b}})Vincent-Cuaz, Vayer, Flamary,
  Corneli, and Courty]{vincent2021online}
C{\'e}dric Vincent-Cuaz, Titouan Vayer, R{\'e}mi Flamary, Marco Corneli, and
  Nicolas Courty.
\newblock Online graph dictionary learning.
\newblock \emph{arXiv preprint arXiv:2102.06555}, 2021{\natexlab{b}}.

\bibitem[Vinh et~al.(2010)Vinh, Epps, and Bailey]{vinh2010information}
Nguyen~Xuan Vinh, Julien Epps, and James Bailey.
\newblock Information theoretic measures for clusterings comparison: Variants,
  properties, normalization and correction for chance.
\newblock \emph{The Journal of Machine Learning Research}, 11:\penalty0
  2837--2854, 2010.

\bibitem[Wang \& Bertsekas(2016)Wang and Bertsekas]{wang2016stochastic}
Mengdi Wang and Dimitri~P Bertsekas.
\newblock Stochastic first-order methods with random constraint projection.
\newblock \emph{SIAM Journal on Optimization}, 26\penalty0 (1):\penalty0
  681--717, 2016.

\bibitem[Wen et~al.(2022)Wen, Ding, Jin, Wang, Xie, and Tang]{wen2022graph}
Hongzhi Wen, Jiayuan Ding, Wei Jin, Yiqi Wang, Yuying Xie, and Jiliang Tang.
\newblock Graph neural networks for multimodal single-cell data integration.
\newblock In \emph{Proceedings of the 28th ACM SIGKDD Conference on Knowledge
  Discovery and Data Mining}, pp.\  4153--4163, 2022.

\bibitem[Xu et~al.(2019{\natexlab{a}})Xu, Luo, and Carin]{xu2019scalable}
Hongteng Xu, Dixin Luo, and Lawrence Carin.
\newblock Scalable gromov-wasserstein learning for graph partitioning and
  matching.
\newblock \emph{Advances in neural information processing systems},
  32:\penalty0 3052--3062, 2019{\natexlab{a}}.

\bibitem[Xu et~al.(2019{\natexlab{b}})Xu, Luo, Zha, and Carin]{xu2019gromov}
Hongteng Xu, Dixin Luo, Hongyuan Zha, and Lawrence Carin.
\newblock Gromov-wasserstein learning for graph matching and node embedding.
\newblock In \emph{International conference on machine learning}, pp.\
  6932--6941. PMLR, 2019{\natexlab{b}}.

\bibitem[Xu et~al.(2021)Xu, Luo, Carin, and Zha]{xu2021learning}
Hongteng Xu, Dixin Luo, Lawrence Carin, and Hongyuan Zha.
\newblock Learning graphons via structured gromov-wasserstein barycenters.
\newblock In \emph{Proceedings of the AAAI Conference on Artificial
  Intelligence}, volume~35, pp.\  10505--10513, 2021.

\bibitem[Xu et~al.(2022)Xu, Liu, Luo, and Carin]{xu2022representing}
Hongteng Xu, Jiachang Liu, Dixin Luo, and Lawrence Carin.
\newblock Representing graphs via gromov-wasserstein factorization.
\newblock \emph{IEEE Transactions on Pattern Analysis and Machine
  Intelligence}, 2022.

\bibitem[Zhang \& Luo(2022)Zhang and Luo]{zhang2020global}
Jiawei Zhang and Zhi-Quan Luo.
\newblock A global dual error bound and its application to the analysis of
  linearly constrained nonconvex optimization.
\newblock \emph{SIAM Journal on Optimization}, 32\penalty0 (3):\penalty0
  2319--2346, 2022.

\bibitem[Zhang et~al.(2021)Zhang, Tong, Jin, Xia, and Guo]{zhang2021balancing}
Si~Zhang, Hanghang Tong, Long Jin, Yinglong Xia, and Yunsong Guo.
\newblock Balancing consistency and disparity in network alignment.
\newblock In \emph{Proceedings of the 27th ACM SIGKDD Conference on Knowledge
  Discovery \& Data Mining}, pp.\  2212--2222, 2021.

\bibitem[Zhou \& So(2017)Zhou and So]{zhou2017unified}
Zirui Zhou and Anthony Man-Cho So.
\newblock A unified approach to error bounds for structured convex optimization
  problems.
\newblock \emph{Mathematical Programming}, 165\penalty0 (2):\penalty0 689--728,
  2017.

\end{thebibliography}
\bibliographystyle{iclr2023_conference}

\newpage
\appendix
Overall, the norm $\| \cdot \|$ means $\ell_2$ norm. Both $\dist2$ and $\proj$ are defined w.r.t. $\ell_2$ norm, i.e., for any convex set $C$ and a point $x$, 
	\begin{align}
	    \dist2(x, C) &= \min_{y \in C} \|x - y\|_2,\\
	    \proj_C(x) &= \arg\min_{y \in C} \|x - y\|_2.
	\end{align}
	\section{Problem Properties of \eqref{eq:gw_qua}}
In this paper, we consider the discrete case, where only $m,n$ samples are drawn from real distributions, i.e., $\mu = \sum_{i=1}^n \mu_i \delta_{x_i}$ and  $\nu = \sum_{j=1}^m \nu_j \delta_{y_j}$. The Gromov-Wasserstein distance between $\mu$ and $\nu$ can be recast into the following minimization problem:
\begin{align*}
& \min_{ \pi \in \Pi(\mu,\nu)} \iint |d_X(x,x')-d_Y(y,y')|^2 d\pi(x,y)d\pi(x',y')\\
\Rightarrow & \min_{ \pi \in \Pi(\mu,\nu)} \sum_{i=1}^n \sum_{j=1}^n \sum_{k=1}^m \sum_{l=1}^m |d_X(x_i,x_j)-d_Y(y_k,y_l)|^2 \pi_{ik}\pi_{jl},
\end{align*}
where $\Pi(\mu,\nu)$ is a generalized Birkhoff polytope, i.e., $\Pi(\mu,\nu) = \{\pi \in  \mathbb{R}_+^{n\times m}| \pi \mathbf{1}_m = \mu,\pi^T\mathbf{1} = \nu \}$.  To further simplify the problem and detect the hidden structures, we reformulate it as a compact form, 
\begin{align*}
& \sum_{i=1}^n \sum_{j=1}^n \sum_{k=1}^m \sum_{l=1}^m (D_X(i,j) - D_Y(k,l))^2 \pi_{ik}\pi_{jl}\\
 =& \sum_{i=1}^n \sum_{j=1}^n \sum_{k=1}^m \sum_{l=1}^m (D_X^2(i,j)+D_Y^2(k,l))\pi_{ik}\pi_{jl} -2\sum_{i=1}^n \sum_{j=1}^n \sum_{k=1}^m \sum_{l=1}^m D_X(i,j)D_Y(k,l)\pi_{ik}\pi_{jl} \\
= & \underbrace{\sum_{i=1}^n \sum_{j=1}^n  D_X^2(i,j) \mu_{i}\mu_{j}+ \sum_{k=1}^m \sum_{l=1}^m  D_Y^2(k,l)\nu_{k}\nu_l}_{\text{constant}} - 2 \text{Tr}(D_X\pi D_Y\pi^T).
\end{align*}
where $D_X(i,j) = d_X(x_i,x_j)$ and $D_Y(k,l) =d_Y(y_k,y_l)$ for $i,j \in [n]$ and $k,l\in[m]$. Thus, \eqref{eq:gw_qua} is obtained. 
\paragraph{Proof of Proposition \ref{prop:erb} --- Luo-Tseng Error Bound Condition for the GW Problem \eqref{eq:gw_qua}}
\begin{proof}
$D_X$ and $D_Y$ are two symmetric matrices and $C_1\cap C_2$ is a convex polyhedral set. By invoking Theorem 2.3 in [24],  the Luo-Tseng local error bound condition \eqref{eq:erb} holds only for the feasible set $C_1 \cap C_2$. That is, 
\begin{equation}
\label{eq:erb_a}
\dist2(\pi,\mathcal{X}) \leq \tau \left\|\pi - \proj_{C_1\cap C_2}(\pi+D_X \pi D_Y)\right\|,
\end{equation}
where $\pi \in C_1\cap C_2$.
Then, we aim at extending \eqref{eq:erb_a} to the whole space. Define $\tilde{\pi} = \proj_{C_1 \cap C_2} (\pi)$, we have 
\begin{equation}\label{ineq:1}
    \dist2(\pi,\mathcal{X}) \leq  \|\pi - \tilde{\pi}\| + d(\tilde{\pi}, \mathcal{X}) \leq  \|\pi - \tilde{\pi}\| + \tau \|\tilde{\pi} -\proj_{C1\cap C_2}(\tilde{\pi}+D_X \tilde{\pi} D_Y)\|
\end{equation}
where the first inequality holds because of the triangle inequality, and the second one holds because of the fact $\tilde{\pi} \in C_1 \cap C_2$ and \eqref{eq:erb_a}. Applying the triangle inequality again, we have
\begin{equation}\label{ineq:2}
\begin{aligned}
    &\|\tilde{\pi} -\proj_{C1\cap C_2}(\tilde{\pi}+D_X \tilde{\pi} D_Y)\| \\ \leq & \|\tilde{\pi} -\proj_{C1\cap C_2}(\pi+D_X \pi D_Y)\| +  \|\proj_{C1\cap C_2}(\pi+D_X \pi D_Y) - \proj_{C1\cap C_2}(\tilde{\pi}+D_X \tilde{\pi} D_Y)\|
    \end{aligned}
\end{equation}
Since $\tilde{\pi} = \proj_{C_1 \cap C_2} (\pi)$ and  the projection operator onto a convex set $\proj_{C_1 \cap C_2}(\cdot)$ is non-expensive, i.e., $\|\proj_{C_1 \cap C_2}(x) - \proj_{C_1 \cap C_2}(y)\| \leq \|x - y\|$ holds for any $x$ and $y$, we have 
\begin{equation}\label{ineq:3}
\begin{aligned}
     \|\tilde{\pi} -\proj_{C1\cap C_2}(\pi+D_X \pi D_Y)\| & =  \|\proj_{C_1 \cap C_2} (\pi) -\proj_{C1\cap C_2}( \proj_{C1\cap C_2}(\pi+D_X \pi D_Y))\| \\
     &\leq \|\pi -\proj_{C1\cap C_2}(\pi+D_X \pi D_Y)\|,
    \end{aligned}
\end{equation}
and 
\begin{equation}\label{ineq:4}
\begin{aligned}
    \|\proj_{C1\cap C_2}(\pi+D_X \pi D_Y) - \proj_{C1\cap C_2}(\tilde{\pi}+D_X \tilde{\pi} D_Y)\| & \leq \| (\pi+D_X \pi D_Y) -  (\tilde{\pi}+D_X \tilde{\pi} D_Y)\| \\ & \leq (\sigma_{\text{max}}(D_X)\sigma_{\text{max}}(D_Y)+1)\|\pi- \tilde{\pi}\|
         \end{aligned}
\end{equation}
where $\sigma_{\text{max}}(D_X)$ and $\sigma_{\text{max}}(D_Y)$ denote the maximum singular value of $D_X$ and $D_Y$, respectively. By substituting \eqref{ineq:2} and \eqref{ineq:4} into \eqref{ineq:1}, we get
\begin{equation}\label{ineq:5}
\begin{aligned}
    \dist2(\pi,\mathcal{X}) & \leq
    ( \tau\sigma_{\text{max}}(D_X)\sigma_{\text{max}}(D_Y)+\tau+1)\|\pi - \tilde{\pi}\|+ \tau\| \pi -\proj_{C_1\cap C_2}(\pi+D_X \pi D_Y) \| \\
    & \leq ( \tau\sigma_{\text{max}}(D_X)\sigma_{\text{max}}(D_Y)+2 \tau+1) \| \pi -\proj_{C_1\cap C_2}(\pi+D_X \pi D_Y) \|
    \end{aligned}
\end{equation}
where the last inequality holds because of the fact that
\begin{align*}
    \|\pi - \tilde{\pi}\| = \|\pi - \proj_{C_1\cap C_2}(\pi)\| = \min_{x \in C_1\cap C_2} \|\pi - x\| \leq \| \pi -\proj_{C_1\cap C_2}(\pi+D_X \pi D_Y) \|.
\end{align*}
By letting $\tau = \tau\sigma_{\text{max}}(D_X)\sigma_{\text{max}}(D_Y)+2\tau+1$, we get the desired result. 
\end{proof}
\section{Convergence Analysis of BAPG}
\begin{assumption}
The critical point set $\mathcal{X}$ is non-empty.
\end{assumption}
\begin{definition}[Bregman Divergence]
We define the proximity measure $D_{h}:\textnormal{dom}(h) \times \textnormal{int}( \operatorname{dom} (h)) \rightarrow \mathbb{R}_{+}$
$$
D_{h}(x, y):=h(x)-h(y)-\langle\nabla h(y), x-y\rangle
$$
The proximity measure $D_{h}$ is the so-called Bregman Distance. It measures the proximity between $x$ and $y$. Indeed, thanks to the gradient inequality, one has
$h$ is convex if and only if $D_{h}(x, y) \geq 0, \forall x \in \operatorname{dom} h, y \in \operatorname{int} \operatorname{dom} h$.
\end{definition}
\textbf{Concrete Examples}
\begin{itemize}
    \item  Relative entropy, or Kullback–Leibler divergence, $h(x) = \sum_ix_i\log x_i$ and the Bregman divergence is given as 
    \[
    D_h(x,y) = \sum_{i}x_i \log\frac{x_i}{y_i}.
    \]
    \item Quadratic function: $h(x) = \frac{1}{2}\|x\|^2$. As such, the Bregman divergence $D_h(x,y)$ will reduce to the Euclidean distance. 
\end{itemize}
For the completeness of our algorithmic development, we also give the details how the general BAPG \eqref{eq:bapg_general} update implies \eqref{eq:bapg_update_c} if we choose the Legendre function as relative entropy. 
\begin{equation*}
\begin{aligned}
\pi^{k+1} & = \mathop{\arg\min}_{\pi \in C_1}\left\{ -\langle \pi, D_Y w^{k}D_Y\rangle + \rho D_h(\pi, w^k)\right\} \\
& =\mathop{\arg\min}_{\pi \in C_1}D_h\left(\pi, w^k\odot \exp\left(\frac{D_X w^{k}D_Y}{\rho}\right)\right), \\
& = \text{diag}(\mu./P^k\mathbf{1}_m)P^k,
\end{aligned}
\end{equation*}
where $P^k = w^k\odot \exp\left(\frac{D_X w^{k}D_Y}{\rho}\right)$.
\begin{equation*}
\begin{aligned}
	w^{k+1} & = \mathop{\arg\min}_{w\in C_2}\left\{ -\langle w,  D_X \pi^{k+1}D_Y\rangle + \rho D_h(w,\pi^{k+1})\right\} \\
	& = \mathop{\arg\min}_{w \in C_2}D_h\left(w, \pi^{k+1} \odot \exp\left(\frac{D_X \pi^{k+1}D_Y}{\rho}\right)\right), \\
	&  = {Z^k} \text{diag}(\nu./{Z^k}^T\mathbf{1}_n),
\end{aligned}
\end{equation*}

where $Z^k = \pi^{k+1}\odot \exp\left(\frac{D_X \pi^{k+1}D_Y}{\rho}\right)$.
The optimality conditions for each iteration follow, 
\begin{equation}
\label{eq:bapg_update}
    \left\{ 
    \begin{aligned}
       &  0 \in - D_X w^{k}D_Y + \rho (\nabla h(\pi^{k+1})-\nabla h(w^k)) + \mathcal{N}_{C_1}(\pi^{k+1}), 
        \\
        & 0 \in - D_X \pi^{k+1}D_Y + \rho (\nabla h(w^{k+1})-\nabla h(\pi^{k+1})) + \mathcal{N}_{C_2}(w^{k+1}). 
    \end{aligned}
    \right. 
\end{equation}
\paragraph{Proof of Lemma \ref{lem:lemma_proj}}
\begin{proof}
To prove 
\begin{equation}
\label{eq:lemma_proj}
\begin{aligned}
&\left\|\proj_{C_{1}}(x)+\proj_{C_{2}}(y)- 2 \proj_{C_1\cap C_2}\left(\frac{x+y}{2}\right)\right\| \leqslant M\left\|\proj_{C_{1}}(x)-\proj_{C_{2}}(y)\right\|,
\end{aligned}
\end{equation}
we first convert the left-hand side of \eqref{eq:lemma_proj} to
 \begin{align*}
    \ & \left\|\proj_{C_{1}}(x)-\proj_{C_1\cap C_2}\left(\frac{x+y}{2}\right) +\proj_{C_{2}}(y)-  \proj_{C_1\cap C_2}\left(\frac{x+y}{2}\right)\right\| \\
    \leq \ & \sqrt{2}\left\| \left(\proj_{C_1\cap C_2}\left(\frac{x+y}{2}\right),\proj_{C_1\cap C_2}\left(\frac{x+y}{2}\right)\right)-(\proj_{C_{1}}(x),\proj_{C_{2}}(y)) \right\|. 
\end{align*}
Then, the inequality can be regarded as the stability of the optimal solution for a linear-quadratic problem, i.e., 
\begin{equation}
\label{eq:proj_cons}
(p(r),q(r)) = 
\begin{aligned}
 & \mathop{\arg\min}_{p,q} \frac{1}{2}\|x-p\|^2 +\frac{1}{2}\|y-q\|^2\\
 & \,\, \quad \text{s.t.} \, \, \quad  p-q= r, \\
 & \quad \quad \quad \quad  p \in C_1, q\in C_2. 
\end{aligned}
\end{equation}
When $r = 0$, the pair $(p(0),q(0))$ satisfies $p(0)=q(0)=\proj_{C_1\cap C_2}\left(\frac{x+y}{2}\right)$. Moreover, the parameter $r$ itself can be viewed as the perturbation quantity, which is indeed the right-hand side of \eqref{eq:lemma_proj}. By invoking Theorem 4.1 in \citep{zhang2020global}, we can bound the distance between two optimal solutions by the perturbation quantity $r$, i.e., 
\begin{align*}
     \|(p(0),q(0))-(p(r),q(r))\| \leq M \|r\|.
\end{align*}

\begin{remark}
For any vector $r$, let $(p(r), q(r))$ be the optimal solution to \eqref{eq:proj_cons}. On the one hand, choosing $r=0$, it is easy to see that $p(0)= q(0) = \proj_{C_1 \cap C_2} (\frac{x+y}{2})$. On the other hand, by choosing $r = \proj_{C_1}(x) - \proj_{C_2}(y)$, it is easy to see that $(p(r), q(r)) = (\proj_{C_1}(x), \proj_{C_2}(y))$. Since Theorem 4.1 in \citep{zhang2020global} implies that for any $r$, we have
	    \begin{align*}
	        \|(p(0), q(0)) - (p(r), q(r))\|_2 \leq M \|r\|_2
	    \end{align*}
	    where $M$ only depends on $C_1$ and $C_2$. By choosing $r = \proj_{C_1}(x) - \proj_{C_2}(y)$, we have
	    \begin{align*}
	        \|(\proj_{C_1 \cap C_2} (\frac{x+y}{2}), \proj_{C_1 \cap C_2} (\frac{x+y}{2})) - (\proj_{C_1}(x), \proj_{C_2}(y)) \leq M \|\proj_{C_1}(x) - \proj_{C_2}(y)\|_2.
	    \end{align*}
	\end{remark}
\end{proof}
\paragraph{Proof of Proposition \ref{prop:bapg_approx} --- Approximation bound of fixed points of BAPG}
\begin{definition}[\textbf{Bounded Linear Regularity} (cf. Definition 5.6 in~\cite{bauschke1996projection})]	\label{def:blr}
	Let $C_1,\ldots,C_N$ be closed convex subsets of $\mathbb{R}^n$ with a non-empty intersection $C$. We say that the collection $\{C_1,\ldots,C_N\}$ is \emph{bounded linearly regular (BLR)} if for every bounded subset $\mathbb{B}$ of $\mathbb{R}^n$, there exists a constant $\kappa>0$ such that 
	\[ d(x,C) \leq \kappa \max_{i\in\{1,\ldots, N\}} d(x,C_i), ~\text{for all} ~x\in \mathbb{B}.\]
\end{definition}
If all of $C_i$ are polyhedral sets, the BLR condition will automatically hold. 
\begin{proof}
Recall the fixed-point set of BAPG:
\begin{equation}
\label{eq:x_bapg2}
    \mathcal{X}_{\text{BAPG}} = \left\{(\pi^\star,w^\star):
\begin{aligned}
& \nabla f(w^\star) + \rho (\nabla h(\pi^\star)-\nabla h(w^\star)) + p =0, p \in \mathcal{N}_{C_1}(\pi^\star) \\ 
& \nabla f(\pi^\star) + \rho (\nabla h(w^\star)-\nabla h(\pi^\star)) + q =0, q \in \mathcal{N}_{C_2}(w^\star)
\end{aligned}
\right\}. 
\end{equation} 

Define $\hat{\pi} = \proj_{C_1\cap C_2}(\pi^\star)$, we first want to argue the following inequality holds, 
\begin{align*}
    \|\hat{\pi}-\pi^\star\| + \|\hat{\pi}-w^\star\| \leq (2\kappa+1)\|\pi^\star-w^\star\|. 
\end{align*}
As the BLR condition is satisfied for the polyhedral constraint, see Definition \ref{def:blr} for details, we have 
\begin{equation}
\label{eq:blr_inf}
\begin{aligned}
    \|\hat{\pi}-\pi^\star\| +  \|\hat{\pi}-w^\star\| & \leq 2\|\hat{\pi}-\pi^\star\| + \|\pi^\star-w^\star\|\\
    & = 2\dist2(\pi^\star,C_1\cap C_2)+\|\pi^\star-w^\star\| \\
    & \leq (2\kappa +1)\|\pi^\star-w^\star\|. 
\end{aligned}
\end{equation}
Based on the stationary points defined in \eqref{eq:x_bapg}, we have, 
\[
\nabla f(w^\star)^T(\hat{\pi}-\pi^\star) + \rho (\nabla h(\pi^\star)-\nabla h(w^\star))^T(\hat{\pi}-\pi^\star)+ p^T(\hat{\pi}-\pi^\star) =0, p \in \mathcal{N}_{C_1}(\pi^\star),
\]
\[
\nabla f(\pi^\star)^T(\hat{\pi}-w^\star) + \rho (\nabla h(w^\star)-\nabla h(\pi^\star))^T(\hat{\pi}-w^\star)  + q^T(\hat{\pi}-w^\star)  =0, q \in \mathcal{N}_{C_2}(w^\star).
\]
Summing up the above two equations, 
\begin{align*}
& \nabla f(w^\star)^T(\hat{\pi}-\pi^\star) + \nabla f(\pi^\star)^T(\hat{\pi}-w^\star) + \rho (\nabla h(\pi^\star)-\nabla h(w^\star))^T(w^\star-\pi^\star)+p^T(\hat{\pi}-\pi^\star)+q^T(\hat{\pi}-w^\star) =0 \\
\stackrel{(\clubsuit)}{\Rightarrow} & 
\nabla f(w^\star)^T(\hat{\pi}-\pi^\star) + \nabla f(\pi^\star)^T(\hat{\pi}-w^\star) - \rho(D_h(\pi^\star,w^\star)+D_h(w^\star,\pi^\star))+p^T(\hat{\pi}-\pi^\star)+q^T(\hat{\pi}-w^\star) =0 \\
\stackrel{(\spadesuit)}{\Rightarrow} & \nabla f(w^\star)^T(\hat{\pi}-\pi^\star) + \nabla f(\pi^\star)^T(\hat{\pi}-w^\star) - \rho(D_h(\pi^\star,w^\star)+D_h(w^\star,\pi^\star)) \ge 0,  
\end{align*}
where $(\clubsuit)$ holds as $(\nabla h(x)-\nabla h(y))^T(x-y) = D_h(x,y)+D_h(y,x)$ and $(\spadesuit)$ follows from the property of the normal cone. For instance, since $q \in \mathcal{N}_{C_1}(\pi^\star)$, we have $q^T(\hat{\pi}-\pi^\star)\leq 0$ (i.e., $\hat{\pi} \in C_1$). Therefore, 
\begin{align*}
    \rho(D_h(\pi^\star,w^\star)+D_h(w^\star,\pi^\star)) & \leq  \nabla f(w^\star)^T(\hat{\pi}-\pi^\star) + \nabla f(\pi^\star)^T(\hat{\pi}-w^\star)\\
    & \leq \|\nabla f(w^\star)\|\|\hat{\pi}-\pi^\star\|+\|\nabla f(\pi^\star)\|\|\hat{\pi}-w^\star\| \\
    & \leq L_f(\|\hat{\pi}-\pi^\star\|+\|\|\hat{\pi}-w^\star\|)\\
    & \stackrel{\eqref{eq:blr_inf}}{\leq} (2\kappa+1)L_f\|\pi^\star-w^\star\|.
\end{align*}
 The next to last inequality holds as $f(\cdot)$ is a quadratic function and the effective domain $C_1 \cap C_2$ is bounded. Thus, the norm of its gradient will naturally have a constant upper bound, i.e., $L_f = \sigma_{\text{max}}(D_X)\sigma_{\text{max}}(D_Y)$. As $h$ is a $\sigma$-strongly convex, we have 
\[D_h(\pi^\star,w^\star) \ge \frac{\sigma}{2}\|\pi^\star-w^\star\|^2.\]
Together with this property, we can quantify the infeasibility error $\|\pi^\star-w^\star\|$,
\begin{equation}
    \label{eq:infeas_BAPG}
     \|\pi^\star-w^\star\| \leq \frac{(2\kappa+1)L_f}{\sigma\rho}.
\end{equation}
When $\rho \rightarrow +\infty$, it is easy to observe that the infeasibility error term $\|\pi^\star-w^\star\|$ will shrink to zero. More importantly, if $\pi^\star=w^\star$, then $X_{\text{BAPG}}$ will be identical to $\mathcal{X}$. 
Next, we target at quantifying the approximation gap between the fixed-point set of BAPG and the critical point set of the original problem \eqref{eq:gw_qua}. Upon \eqref{eq:x_bapg2}, we have 
\[
\nabla f \left(\frac{\pi^\star+w^\star}{2}\right)+\frac{p+q}{2} = 0,  
\]
where $\nabla f(\cdot)$ is a linear operator. By applying the Luo-Tseng local error bound condition of \eqref{eq:gw_qua}, i.e., Proposition \ref{prop:erb}, we have
\begin{equation*}
\begin{aligned}
     \dist2\left(\frac{\pi^\star+w^\star}{2}, \mathcal{X}\right)  & \leq \tau \left\|\frac{\pi^\star+w^\star}{2} - \proj_{C_1\cap C_2}\left(\frac{\pi^\star+w^\star}{2}-\nabla f \left(\frac{\pi^\star+w^\star}{2}\right)\right)\right\| \\
     & =  \tau \left\|\frac{\pi^\star+w^\star}{2} - \proj_{C_1\cap C_2}\left(\frac{\pi^\star+p+w^\star+q}{2}\right)\right\| \\
     & \stackrel{(\clubsuit)}{=} \frac{\tau}{2}\left\| \proj_{C_1}(\pi^\star+p) + \proj_{C_1}(w^\star+q) - 2\proj_{C_1\cap C_2}\left(\frac{\pi^\star+p+w^\star+q}{2}\right) \right\| \\
     & \stackrel{(\spadesuit)}{\leq} \frac{M\tau}{2}\|\pi^\star- w^\star\|,
\end{aligned}
\end{equation*}
where $(\clubsuit)$ holds due to the normal cone property, that is, $\proj_{C}(x+z) = x$ for any $x \in C$ and $z \in \mathcal{N}_C(x)$, and $(\spadesuit)$ follows from Lemma \ref{lem:lemma_proj}. Incorporating with \eqref{eq:infeas_BAPG}, the approximation bound for BAPG has been characterized quantitatively, i.e., 
\[
\dist2\left(\frac{\pi^\star+w^\star}{2}, \mathcal{X}\right)\leq  \frac{(2\kappa+1)L_fM}{2\sigma\rho}. 
\]
\end{proof}
\paragraph{Proof of Proposition \ref{prop:basic_bapg} --- Sufficient decrease property of BAPG}
\begin{proof}
We first observe from the optimality conditions of main updates, i.e., 
\begin{align}
\label{eq:bapg_update_opt1}
& 0 \in \nabla_\pi f(\pi^{k},w^{k}) + \rho (\nabla h(\pi^{k+1})-\nabla h(w^{k})) + \partial g_1(\pi^{k+1}) \\
\label{eq:bapg_update_opt2}
&  0 \in \nabla_w f(\pi^{k+1},w^{k})+ \rho (\nabla h(w^{k+1})-\nabla h(\pi^{k+1})) +\partial g_2(w^{k+1})
\end{align}
where $g_1(\pi) = I_{\{\pi \in C_1\}}$ and $g_2(w) = I_{ \{w\in C_2\}}$. 
Due to the convexity of $g_1(\cdot)$, it is natural to obtain, 
\begin{align*}
g_1(\pi^{k})-g_1(\pi^{k+1}) & \ge -\langle \nabla_\pi f(\pi^k, w^{k})
+\rho (\nabla h(\pi^{k+1})-\nabla h(w^{k})), \pi^k-\pi^{k+1}\rangle \\
& = -\langle \nabla_\pi f(\pi^k, w^{k}),\pi^k-\pi^{k+1}\rangle +\langle \rho (\nabla h(\pi^{k+1})-\nabla h(w^{k})), \pi^{k+1}-\pi^{k}\rangle
\end{align*}
As $f(\pi,w) = -\text{tr}(D_X\pi D_Y w^T)$ is a bilinear function, we have,
\[
f(\pi^{k},w^k) - f(\pi^{k+1}, w^k)-\langle \nabla_\pi f(\pi^k, w^{k}),\pi^k-\pi^{k+1}\rangle = 0. 
\]
Consequently, we get 
\begin{equation}
\label{eq:BAPG_up1}
    f(\pi^{k},w^k)+g_1(\pi^{k})-f(\pi^{k+1}, w^k)-g_1(\pi^{k+1}) \ge  \rho \langle \nabla h(\pi^{k+1})-\nabla h(w^{k}), \pi^{k+1}-\pi^{k}\rangle. 
\end{equation}
Similarly, based on the $w$-update, we obtain 
\begin{equation}
\label{eq:BAPG_up2}
    f(\pi^{k+1},w^{k})+g_2(w^{k})-f(\pi^{k+1}, w^{k+1})-g_2(w^{k+1}) \ge  \rho \langle \nabla h(w^{k+1})-\nabla h(\pi^{k+1}), w^{k+1}-w^{k}\rangle. 
\end{equation}
Combine with \eqref{eq:BAPG_up1} and \eqref{eq:BAPG_up2}, we obtain
\begin{equation}
\label{eq:bapg_decrease}
\begin{aligned}
    & f(\pi^{k},w^k)+g_1(\pi^{k})+g_2(w^{k})-f(\pi^{k+1}, w^{k+1})-g_1(\pi^{k+1})-g_2(w^{k+1}) \\
    \geq \, \, &  \rho \langle \nabla h(\pi^{k+1})-\nabla h(w^{k}), \pi^{k+1}-\pi^{k}\rangle + \rho \langle \nabla h(w^{k+1})-\nabla h(\pi^{k+1}), w^{k+1}-w^{k}\rangle. 
\end{aligned}
\end{equation}
The right-hand side can be further simplified, 
\begin{align*}
     & \rho \langle \nabla h(\pi^{k+1})-\nabla h(w^{k}), \pi^{k+1}-\pi^{k}\rangle + \rho \langle \nabla h(w^{k+1})-\nabla h(\pi^{k+1}), w^{k+1}-w^{k}\rangle\\
     \stackrel{(\clubsuit)}{=}& \rho (D_h(\pi^k,\pi^{k+1})+D_h(\pi^{k+1},w^k)-D_h(\pi^k,w^k)) +\rho(D_h(w^k,w^{k+1}) + \\
     & D_h(w^{k+1},\pi^{k+1})-D_h(w^k,\pi^{k+1}))\\
     \stackrel{(\spadesuit)}{=} &  \rho D_h(\pi^k,\pi^{k+1}) + \rho D_h(w^k,w^{k+1}) - \rho D_h(\pi^k,w^k)+ \rho D_h(\pi^{k+1},w^{k+1}).
\end{align*}
Here, $(\clubsuit)$ uses the fact that the three-point property of Bregman divergence holds, i.e., 
 For any $y, z \in$ int $\operatorname{dom} h$ and $x \in \operatorname{dom} h$,
$$
D_{h}(x, z)-D_{h}(x, y)-D_{h}(y, z)=\langle\nabla h(y)-\nabla h(z), x-y\rangle.
$$


By combining with \eqref{eq:BAPG_up1} and \eqref{eq:BAPG_up2}, we obtain \eqref{eq:bapg_decrease}. The right-hand side of \eqref{eq:bapg_decrease} can be further simplified, 
\begin{align*}
     & \rho \langle \nabla h(\pi^{k+1})-\nabla h(w^{k}), \pi^{k+1}-\pi^{k}\rangle + \rho \langle \nabla h(w^{k+1})-\nabla h(\pi^{k+1}), w^{k+1}-w^{k}\rangle\\
     \stackrel{(\clubsuit)}{=}& \rho (D_h(\pi^k,\pi^{k+1})+D_h(\pi^{k+1},w^k)-D_h(\pi^k,w^k)) +\rho(D_h(w^k,w^{k+1}) + \\
     & D_h(w^{k+1},\pi^{k+1})-D_h(w^k,\pi^{k+1}))
\end{align*}
Here, $(\clubsuit)$ uses the fact that the three-point property of Bregman divergence holds, i.e., 
for any $y, z \in$ int $\operatorname{dom} h$ and $x \in \operatorname{dom} h$,
$$
D_{h}(x, z)-D_{h}(x, y)-D_{h}(y, z)=\langle\nabla h(y)-\nabla h(z), x-y\rangle.
$$
This together with \eqref{eq:bapg_decrease} implies
\begin{equation}
\label{eq:suff_bapg_new}
F_\rho(\pi^{k+1},w^{k+1}) - F_\rho(\pi^k ,w^k) \leq -  \rho D_h(\pi^k,\pi^{k+1}) - \rho D_h(w^k,w^{k+1}) -\rho \left(D_h(\pi^{k+1},w^k) - D_h(w^k, \pi^{k+1})\right) . 
\end{equation}

Summing up \eqref{eq:suff_bapg_new} from $k=0$ to $+\infty$, we obtain 
\[
F_\rho(\pi^{\infty},w^\infty)-F_\rho(\pi^0,w^0) + \rho \sum_{k=0}^\infty \left(D_h(\pi^{k+1},w^k) - D_h(w^k, \pi^{k+1})\right) \leq -\rho \sum_{k=0}^\infty \left(D_h(\pi^k,\pi^{k+1}) + D_h(w^k,w^{k+1}) \right). 
\]

As the potential function $F_\rho(\cdot,\cdot)$ is coercive and $\{(\pi^{k},w^k)\}_{k \ge 0}$ is a bounded sequence, the assumption that $\sum_{k=0}^\infty \left(D_h(\pi^{k+1},w^k) - D_h(w^k, \pi^{k+1})\right)$ is bounded implies that the left-hand side is bounded, which means
\[
\sum_{k=0}^\infty \left(D_h(\pi^k,\pi^{k+1}) + D_h(w^k,w^{k+1}) \right) < +\infty. 
\]

\end{proof}
\paragraph{Proof of Theorem \ref{thm:sub_bapg} --- Subsequence convergence result of BAPG}
\begin{proof}
 Let $(\pi^\infty, w^\infty)$ be a limit point of the sequence $\{(\pi^{k},w^{k})\}_{k \ge 0 }$. Then, there exists a sequence $\{n_k\}_{k\ge0}$ such that $\{(\pi^{n_k},w^{n_k})\}_{k \ge 0 }$ converges to $(\pi^\infty, w^\infty)$. Replacing $k$ by $n_k$ in \eqref{eq:bapg_update_opt1} and \eqref{eq:bapg_update_opt2}, 
 taking limits on both sides as $k\rightarrow\infty$ 
\begin{align*}
& 0 \in \nabla_\pi f(\pi^{\infty},w^{\infty}) + \rho (\nabla h(\pi^{\infty})-\nabla h(w^{\infty})) + \partial g_1(\pi^{\infty}) \\
&  0 \in \nabla_w f(\pi^{\infty},w^{\infty})+ \rho (\nabla h(w^{\infty})-\nabla h(\pi^{\infty})) +\partial g_2(w^{\infty}).
\end{align*}
Based on the fact that $\nabla_\pi f(\pi^{\infty},w^{\infty}) = \nabla f(w^\infty)$ and $\nabla_w f(\pi^{\infty},w^{\infty}) = \nabla f(\pi^\infty)$, it can be easily concluded that $(\pi^\infty, w^\infty) \in \mathcal{X}_{\text{BAPG}}$.
\end{proof}

More importantly, the above subsequence convergence result can be extended to the global convergence if the Legendre function $h(\cdot)$ is quadratic.  
\begin{theorem}[Global Convergence of BAPG --- Quadratic Case]
The sequence $\{(\pi^{k},w^k)\}_{k \ge 0 }$ converges to a critical point of $F_\rho(\pi, w)$. 
\end{theorem}
\begin{proof}
To invoke the Kurdyka-Lojasiewicz analysis framework 
developed in \citep{attouch2010proximal,attouch2013convergence}, we have to establish two crucial properties --- sufficient decrease and safeguard condition. The first one has already been proven in Proposition \ref{prop:basic_bapg}. Then, we would like to show that there exists a constant $\kappa_2$ such that, 
\[
\dist2(0,   \partial F_\rho(\pi^{k+1},w^{k+1})) \leq \kappa_2 (\|\pi^{k+1}-\pi^k\|+\|w^{k+1}-w^k\|).
\]

At first, noting that 
\begin{equation*}
    \partial F_\rho(\pi^{k+1},w^{k+1}) = \left[
    \begin{aligned}
      & \nabla_\pi f(\pi^{k+1},w^{k+1}) + \rho (\pi^{k+1}-w^{k+1}) + \partial g_1(\pi^{k+1}) \\
      & \nabla_w f(\pi^{k+1},w^{k+1}) + \rho (w^{k+1}-\pi^{k+1}) + \partial g_2(w^{k+1})
    \end{aligned}
    \right].
\end{equation*}
Again, together with the updates \eqref{eq:bapg_update_opt1} and \eqref{eq:bapg_update_opt2}, this implies
 \begin{equation*}
 \begin{aligned}
 \dist2^2(\mathbf{0},\partial F_\rho(\pi^{k+1},w^{k+1})) & \leq \|\nabla_\pi f(\pi^{k+1},w^{k+1})-\nabla_\pi f(\pi^{k},w^{k})\|^2 + \rho \| w^k- w^{k+1}\|^2 \\
 & + \|\nabla_w f(\pi^{k+1},w^{k})-\nabla_w f(\pi^{k+1},w^{k+1})\|^2 \\
 & \leq (\sigma^2_{\text{max}}(D_X)\sigma^2_{\text{max}}(D_Y)+\rho) \| w^k- w^{k+1}\|^2, 
 \end{aligned}
 \end{equation*}
 where $\sigma_{\text{max}}(\cdot)$ denotes the maximum singular value. 
 By letting $\kappa_2 = \sqrt{\sigma^2_{\text{max}}(D_X)\sigma^2_{\text{max}}(D_Y)+\rho}$, we get the desired result.
\end{proof}

\newpage 
\section{Additional Experiment Details}
\begin{table}[h]
\caption{Comparison of various algorithms. "Exactly solve " means that  algorithms can reach a critical point of the original GW problem instead of the approximation problem.}
\centering
\begin{tabular}{@{}l|c|c|c@{}}
\toprule
& Single-loop? & Provable? & Exactly solve \eqref{eq:gw_qua}?\,\\
      \midrule
\, eBPG~\citep{solomon2016entropic} &  \textcolor{blue!50!black}{No} & \textcolor{red!50!black}{Yes} &  \textcolor{blue!50!black}{No}\\
\, BPG~\citep{xu2019gromov} & \textcolor{blue!50!black}{No} & \textcolor{red!50!black}{Yes}&\textcolor{red!50!black}{Yes}\\
\, BPG-S~\citep{xu2019scalable} & \textcolor{red!50!black}{Yes} &  \textcolor{blue!50!black}{No}&  \textcolor{blue!50!black}{No}\\
\, FW~\citep{titouan2019optimal}&\textcolor{blue!50!black}{No} & \textcolor{red!50!black}{Yes}&\textcolor{red!50!black}{Yes} \\
\, BAPG (\textbf{our method}) & \textcolor{red!50!black}{Yes}& \textcolor{red!50!black}{Yes}& \textcolor{blue!50!black}{No}\\
      \bottomrule
\end{tabular}
\label{tab:algo}
\end{table}
\begin{table}[h]
\centering
\caption{Statistics of databases for graph alignment.}
\begin{tabular}{c|ccc}
\toprule
Dataset    & \# Samples  & Ave. Nodes  & Ave. Edges \\\midrule
Synthetic   & 300     & 1500        & 56579 \\
Proteins    & 1113    & 39.06       & 72.82 \\
Enzymes     & 600     & 32.63       & 62.14 \\
Reddit      & 500     & 375.9       & 449.3 \\\bottomrule
\end{tabular}

\label{tab:align_data}
\end{table}

\begin{table*}[h]
\centering
\caption{Graph alignment results (Mean $\pm$ Std.) of compared GW-based methods in 5 independent trials over different random seeds in the noise generating process.}
\begin{tabular}{c|cccc}
\toprule
Method   & Synthetic                             & Proteins-Noisy                              & Enzymes-Noisy                               & Reddit-Noisy                                \\
\midrule
FW & 24.50$\pm$0.95                           & 20.24$\pm$0.38                          & 22.80$\pm$1.15                           & 18.34$\pm$0.22                            \\
SpecGW  & 13.27$\pm$0.49                           & 19.31$\pm$0.70                           & 21.14$\pm$0.45                           & 19.66$\pm$0.16                           \\
ScalaGW & 17.93$\pm$0.27                           & 16.05$\pm$0.06                           & 11.46$\pm$0.22                           & 0.70$\pm$0.05                            \\
eBPG     & 34.33$\pm$0.61                           & 45.85$\pm$0.73                           & 60.46$\pm$0.57                           & 3.34$\pm$0.03                            \\
BPG      & 57.56$\pm$0.45                           & 52.46$\pm$0.25                           & 62.32$\pm$0.51                           & 36.68$\pm$0.02                           \\
BPG-S      & 61.48$\pm$0.34                                 & 52.74$\pm$0.13                           & 62.21$\pm$0.55                           & 36.68$\pm$0.02                  \\
KL-BAPG    & \textbf{99.79$\pm$0.21} & \textbf{57.16$\pm$0.23} & \textbf{62.85$\pm$0.48} & \textbf{49.45$\pm$0.02}\\
\bottomrule
\end{tabular}
\end{table*}

\textcolor{black}{
\subsection{Verification of the Boundedness of Accumulative Asymmetric Error for 2D Toy Example in Sec \ref{sec:toy}}\label{sec:verify}}
\vspace{-4mm}
\begin{figure}[b!]
    \centering
    \includegraphics[width=0.6\textwidth]{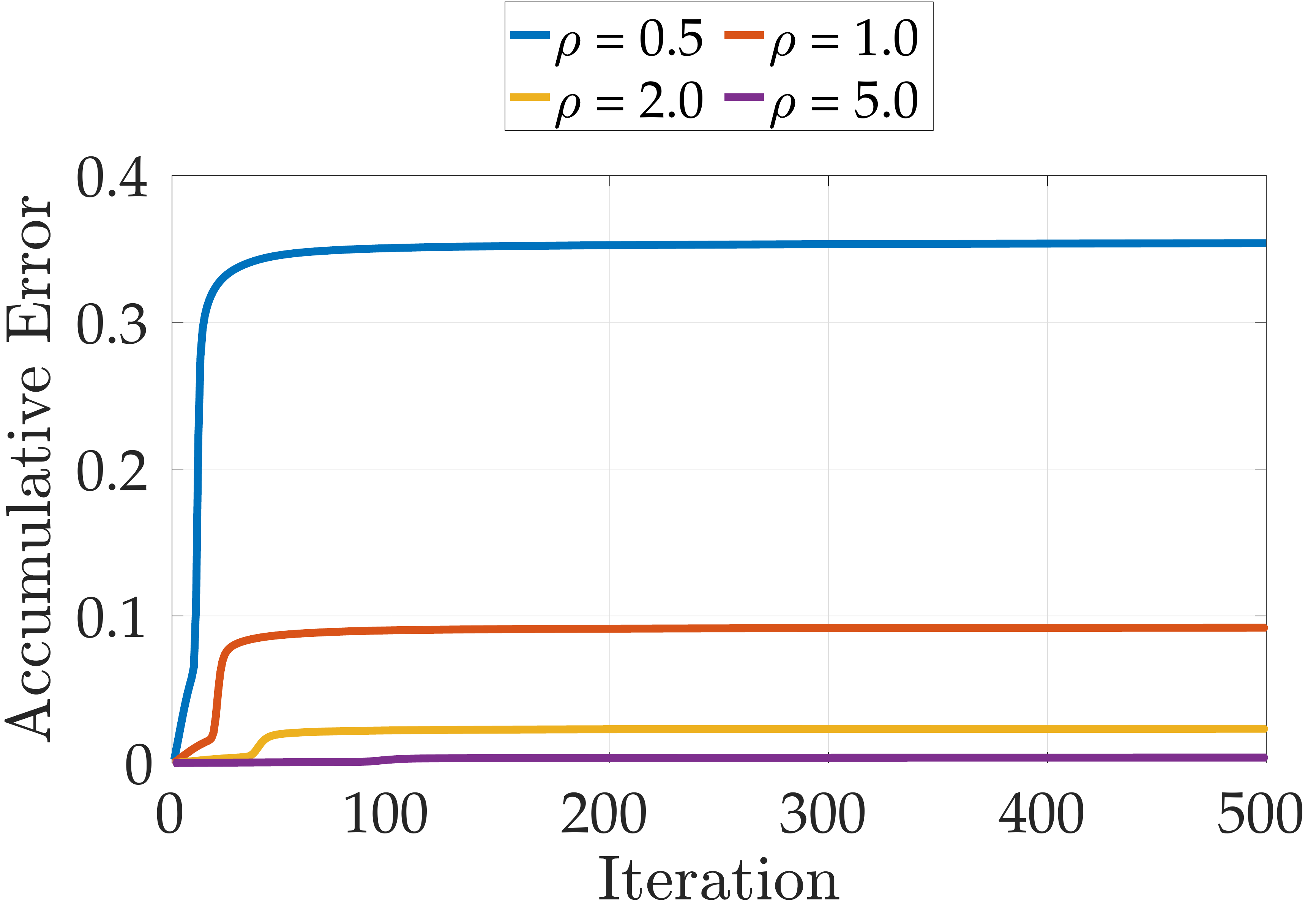}
    \caption{Accumulative Asymmetric Error}
    \label{fig:accumulative_error}
\end{figure}

\subsection{Additional Experiments on Trade-off among Efficiency, Accuracy and Feasibility }
In this subsection, we provide comprehensive experiments to demonstrate the effectiveness of KL-BAPG. Table \ref{tab:mar_compare} shows that FW and SpecGW use the LP as the inner solver, resulting in almost zero infeasibility errors; eBPG and BPG use the sinkhorn solver, resulting in moderate infeasibility errors; BPG-S and the proposed KL-BAPG are infeasible methods, resulting in relatively larger infeasibility errors, but enjoy faster convergence speed. Note that BPG-S is a heuristic method and lacks theoretical support, and may not perform well (or even converge) under naturally noisy observations. 
Overall, we can see that the proposed KL-BAPG is able to sacrifice some feasibility to gain both computational efficiency and matching accuracy.

\begin{table}[h]
\centering
\caption{The average marginal error $\|\pi^T\mathbf{1}_n -\nu\|+\|\pi \mathbf{1}_m -\mu\|$ of compared baselines on the raw version of four graph alignment databases.}
\resizebox{\linewidth}{!}{
\begin{tabular}{c|crc|crc|crc|crc}\toprule
\multirow{2}{*}{Method} & \multicolumn{3}{c|}{Synthetic}                              & \multicolumn{3}{c|}{Proteins}                               & \multicolumn{3}{c|}{Enzymes}                               & \multicolumn{3}{c}{Reddit}                                \\
                        & Acc            & Time          & Error                     & Raw            & Time          & Error                     & Raw            & Time         & Error                     & Raw            & Time         & Error                     \\\midrule
FW                      & 24.50          & 8182  & 4.6e-9           & 29.96          & 54.2  & 4.5e-9           & 32.17          & 10.8  & \textless{}1e-10 & 21.51          & 1121         & \textless{}1e-10 \\
SpecGW                  & 13.27          & 1462  & \textless{}1e-10 & 78.11          & 30.7  & \textless{}1e-10 & 79.07          & 6.7   & \textless{}1e-10 & 50.71          & 1074         & \textless{}1e-10 \\
eBPG                    & 34.33          & 9502  & 3.0e-10          & 67.48          & 208.2 & 9.2e-6           & 78.25          & 499.7 & 1.2e-6           & 3.76           & 1234         & 2.9e-6           \\
BPG                     & 57.56          & 22600 & 2.5e-7           & 71.99          & 130.4 & 2.2e-7           & 79.19          & 73.1  & 1.2e-7           & 39.04          & 1907         & 1.3e-6           \\
BPG-S                   & 61.48          & 18587 & 9.9e-6           & 71.74          & 40.4  & 5.5e-6           & 79.25          & 13.4  & 1.2e-6           & 39.04          & 1431         & 7.8e-6           \\
\bottomrule
\end{tabular}
}\label{tab:mar_compare}
\end{table}

We also conduct sensitivity experiments of KL-BAPG with respect to the step size $\rho$. The results in Table \ref{tab:mar_compare2} corroborate Proposition \ref{prop:bapg_approx} empirically and guide us on how to choose the step size. The larger $\rho$ means a smaller infeasibility error but a slower convergence speed. If the resulting infeasibility error is not too large (i.e., the step size $\rho$ is not too small), the final matching accuracy is robust to the step size and only affects the convergence speed. To further demonstrate that the moderate infeasibility error is not a big deal for the graph alignment task, we execute the post-processing step (i.e., run additional 100 steps with a large step size) to make the infeasibility error smaller. The matching accuracy does not change.

\begin{table}[h]
\caption{The average marginal error of KL-BAPG with different $\rho$ on four graph alignment databases. * represents further running KL-BAPG with a larger $\rho$ (i.e., $\rho=100$) for 100 epochs after convergence.}
\resizebox{\linewidth}{!}{
\begin{tabular}{c|crc|crc|crc|crc}\toprule
\multirow{2}{*}{KL-BAPG} & \multicolumn{3}{c|}{Synthetic} & \multicolumn{3}{c|}{Proteins} & \multicolumn{3}{c|}{Enzymes} & \multicolumn{3}{c}{Reddit} \\
                      & Acc      & Time    & Error    & Raw      & Time   & Error    & Raw     & Time   & Error    & Raw     & Time   & Error   \\\midrule
$\rho$=0.5            & 94.58    & 4610    & 9.4e-4   & 77.21    & 88.7   & 8.8e-6   & 79.78   & 18.2   & 8.6e-7   & 50.19   & 264    & 3.6e-6  \\
$\rho$=0.5*           & 94.57    & 4703    & 1.4e-5   & 77.22    & 98.1   & 8.7e-8   & 79.79   & 22.2   & 4.8e-8   & 50.19   & 280    & 8.7e-8  \\\midrule
$\rho$=0.2            & 99.50    & 2760    & 2.0e-3   & 78.27    & 63.5   & 9.4e-5   & 79.39   & 10.4   & 6.2e-6   & 50.92   & 197    & 1.9e-6  \\
$\rho$=0.2*           & 99.50    & 2854    & 1.5e-5   & 78.27    & 72.9   & 2.1e-7   & 79.39   & 14.3   & 5.6e-8   & 50.92   & 213    & 9.1e-8  \\\midrule
$\rho$=0.1            & 99.79    & 1253    & 3.3e-3   & 78.18    & 59.1   & 1.4e-3   & 79.67   & 6.9    & 2.1e-5   & 50.93   & 115    & 5.1e-6  \\
$\rho$=0.1*           & 99.79    & 1349    & 9.8e-6   & 78.16    & 68.6   & 7.8e-6   & 79.67   & 10.8   & 6.1e-8   & 50.93   & 130    & 9.5e-8  \\\midrule
$\rho$=0.05           & 99.20    & 956     & 5.2e-3   & 77.59    & 36.2   & 1.7e-2   & 79.27   & 5.0    & 2.9e-3   & 50.96   & 102    & 1.5e-5  \\
$\rho$=0.05*          & 99.20    & 1051    & 8.1e-6   & 77.48    & 45.7   & 2.1e-5   & 79.30   & 8.9    & 3.1e-6   & 50.96   & 118    & 9.7e-8  \\\midrule
$\rho$=0.01           & 93.63    & 616     & 4.1e-2   & 60.08    & 14.6   & 1.3e-1   & 66.86   & 4.0    & 8.4e-2   & -       & -      & -       \\
$\rho$=0.01*          & 95.02    & 712     & 4.9e-6   & 60.92    & 24.0   & 9.3e-3   & 67.20   & 7.8    & 4.4e-5   & -       & -      & -  \\\bottomrule    
\end{tabular}
}\label{tab:mar_compare2}
\end{table}

\newpage
\subsection{Multi-Omics Single-Cell Integration }
To validate the application capability of KL-BAPG beyond graph alignment and partition, we further evaluate the proposed BAPG on an additional task: multi-omics single-cell integration. This task aims to integrate multiple molecular features at different modalities in a cell, e.g., gene expressions and chromatin accessibility \citep{teichmann2020method,wen2022graph,cao2020unsupervised,cao2022manifold}. It offers opportunities for gaining holistic views of cells and was selected as ``Method of the Year 2019'' by Nature \citep{teichmann2020method}.

\paragraph{Task Definition}Suppose we have two single cell multi-omics datasets, $X = [x_1;\cdots;x_{n_x}] \in \mathbb R^{n_x \times d_x}$ and $ Y = [y_1;\cdots;y_{n_y}] \in \mathbb R^{n_x \times d_x}$, across two modalities, where $n_x/n_y$ and $d_x/d_y$ are the number of cell samples and feature dimensions for $X$/$Y$, respectively. Given $X$ and $Y$, the task of multi-omics single-cell integration aims at calculating the joint representation of two modalities. In the joint representation, the same cell types in $X$ and $Y$ should be clustered together. Thus, aligning the same types of cells across $X$ and $Y$ in advance can improve the quality of the joint representation.

\begin{table}[h!]
\caption{Statistics of two real-world single-cell multi-omics datasets.}\label{tab:omics_data}
\centering
\begin{tabular}{l|ll|lll}\toprule

                               & \multicolumn{2}{c|}{scGEM} & \multicolumn{3}{c}{scMNT}               \\\midrule
\multirow{2}{*}{Modality Name} & Gene       & DNA          & Gene      & DNA         & Chromatin     \\
                               & Epression  & Methylation  & Epression & Methylation & Accessibility \\\midrule
Samples $n$                    & 177        & 177          & 612       & 709         & 1940          \\
Features $d$                   & 34         & 27           & 300       & 300         & 300          \\\bottomrule
\end{tabular}
\end{table}

\paragraph{Dataset and Evaluation} Following \citet{cao2020unsupervised,cao2022manifold}, we conduct experiments on two real-world single-cell multi-omics datasets: (1) the scGEM dataset for the single-cell analysis of genotype, expression and methylation data and (2) the scNMT dataset for the single-cell analysis of nucleosome, methylome and transcriptome data. The dataset statistics is listed in Table \ref{tab:omics_data}. We use the label transfer accuracy (i.e., the percentage of aligned cell pairs with the same cell type across two modalities) to measure the ability to transfer labels of the shared cells from one modality to another.

\paragraph{Baselines} We compare the proposed KL-BAPG with other GW methods in our paper by directly setting the distance matrix $D_X$ and $D_Y$ as the feature similarity matrix. Besides, we add two baselines specified for this task, including UnionCom \citep{cao2020unsupervised} and Pamona \citep{cao2022manifold}. Pamona also relies on calculating the GW distance method and uses eBPG as the GW solver. We replace eBPG in Panona with KL-BAPG and obtain a new method named Pamona-BAPG.

\begin{table}[h!]
\caption{Comparison of the label transfer accuracy (\%) and wall-clock time (seconds) on the scGEM and scNMT datasets for single-cell multi-omics integration. D., G., and Ch. represent the modalities of DNA Methylation, Gene Expression, and Chromatin Accessibility, respectively.}\label{tab:omics_result}

\begin{tabular}{c|lll|lllll}
\toprule
\multicolumn{1}{l}{} & \multicolumn{3}{|c|}{scGEM}                                                              & \multicolumn{5}{c}{scMNT}                                                                                                                              \\\midrule
Method               & \multicolumn{1}{c}{D.→G.} & \multicolumn{1}{c}{G.→D.} & \multicolumn{1}{c|}{Time} & \multicolumn{1}{c}{D.→G.} & \multicolumn{1}{c}{G.→D.} & \multicolumn{1}{c}{Ch.→G.} & \multicolumn{1}{c}{G→Ch.} & \multicolumn{1}{c}{Time} \\\midrule
FW                   & 40.7                         & 40.7                         & 0.7                      & 74.2                         & \textbf{74.3}                & 44.1                          & 45.4                          & 19.9                     \\
eBPG                 & 7.9                          & 10.7                         & 7.4                      & 63.4                         & 68.0                         & 46.6                          & 44.8                          & 64.2                     \\
specBPG              & 21.5                         & 21.5                         & 0.4                      & 61.4                         & 57.2                         & 36.4                          & 39.7                          & 2.2                      \\
BPG-S                & 9.0                          & 14.1                         & 1.2                      & 69.0                         & 73.2                         & 60.1                          & 60.7                          & 2.3                      \\
BPG                  & 11.9                         & 11.3                         & 11.5                     & \textbf{75.1}                & 73.2                         & 60.8                          & 60.0                          & 22.3                     \\
KL-BAPG                 & \textbf{52.0}                & \textbf{50.8}                & \textbf{0.3}             & 74.0                         & \textbf{74.3}                & \textbf{68.0}                 & \textbf{67.0}                 & \textbf{0.8}             \\\midrule
UnionCom             & 48.0                         & 50.8                         & 6.4                      & 85.6                         & 76.3                         & 80.1                          & 78.8                          & 147.4                    \\
Pamona               & 61.6                         & \textbf{58.2}                & 4.4                      & 71.6                         & 56.6                         & 71.9                          & 66.8                          & 26.1                     \\
Pamona-BAPG          & \textbf{67.2}                & \textbf{58.2}                & 0.7                      & \textbf{85.7}                & \textbf{91.3}                & \textbf{80.7}                 & \textbf{82.9}                 & 6.5      \\\bottomrule               
\end{tabular}
\end{table}

\paragraph{Results of All Methods} Table \ref{tab:omics_result} summarizes the results of all compared methods on the scGWM and scNMT datasets. Among all GW methods, KL-BAPG achieves the best performance on 4 out of 6 scores and has the shortest running time. UnionCom and Pamona are better than most pure GW methods, as they involve additional pre-processing and post-processing techniques specific to this task. Moreover, Pamona-BAPG outperforms Pamona on 5 out of 6 scores and reduces the running time by 16\%-25\%. This indicates that the calculation of GW distance is a bottleneck for this task and a better GW solver can greatly facilitate related downstream applications. As this is just a preliminary study to show the potential of applying KL-BAPG to other tasks, we leave other potential application-driven tasks for future work.

\subsection{Source codes of all baselines used in this paper and their licenses}
\begin{itemize}
    \item BPG-S (GPL-3.0 license)~\citep{xu2019gromov}: \url{https://github.com/HongtengXu/gwl}
    \item ScalaGW (No license)~\citep{xu2019scalable}: \url{https://github.com/HongtengXu/s-gwl}
    \item SpecGW (MIT license)~\citep{chowdhury2021generalized}:\\ \url{https://github.com/trneedham/Spectral-Gromov-Wasserstein}
    \item eBPG (MIT license)~\citep{flamary2021pot}:  \url{https://github.com/PythonOT/POT} 
    \item FW (MIT license)~\citep{flamary2021pot}: \url{https://github.com/PythonOT/POT} 
    \item IPFP, RRWM, SpecMethod (Mulan PSL v2 lincense): \url{https://github.com/Thinklab-SJTU/pygmtools}
\end{itemize}

\end{document}